\documentclass[conference]{IEEEtran}
\IEEEoverridecommandlockouts
\usepackage{booktabs}
\usepackage{soul}
\usepackage{framed}
\usepackage{subfigure}
\usepackage[linesnumbered,noline,noend,boxed]{algorithm2e}

\usepackage{amssymb}
\usepackage{mathtools}
\usepackage{listings}
\usepackage{courier}
\usepackage{xcolor}
\usepackage{amsthm}
\usepackage{enumitem}
\newtheorem{lemma}{Lemma}
\newtheorem{theorem}{Theorem}
\newtheorem{definition}{Definition}

\def\BibTeX{{\rm B\kern-.05em{\sc i\kern-.025em b}\kern-.08em
    T\kern-.1667em\lower.7ex\hbox{E}\kern-.125emX}}
\begin{document}

\title{Privacy Preserving Distributed Energy Trading\\
{\small\textsuperscript{*}In Proceedings of the 40th International Conference on Distributed Computing Systems, 2020}
}

\author{Shangyu Xie\\
\IEEEauthorblockA{\textit{Department of Computer Science} \\
\textit{Illinois Institute of Technology}\\
sxie14@hawk.iit.edu}
\and
\IEEEauthorblockN{Han Wang}
\IEEEauthorblockA{\textit{Department of Computer Science} \\
\textit{Illinois Institute of Technology}\\
hwang185@hawk.iit.edu}
\and
\IEEEauthorblockN{Yuan Hong}
\IEEEauthorblockA{\textit{Department of Computer Science} \\
\textit{Illinois Institute of Technology}\\
yuan.hong@iit.edu}
\and
\IEEEauthorblockN{My Thai}
\IEEEauthorblockA{\textit{Department of CISE} \\
\textit{University of Florida}\\
mythai@cise.ufl.edu}}

\maketitle

\begin{abstract}
The smart grid incentivizes distributed agents with local generation (e.g., smart homes, and microgrids) to establish multi-agent systems for enhanced reliability and energy consumption efficiency. Distributed energy trading has emerged as one of the most important multi-agent systems on the power grid by enabling agents to sell their excessive local energy to each other or back to the grid. However, it requests all the agents to disclose their sensitive data (e.g., each agent's fine-grained local generation and demand load). In this paper, to the best of our knowledge, we propose the first privacy preserving distributed energy trading framework, \emph{Private Energy Market} (PEM), in which all the agents privately compute an optimal price for their trading (ensured by a Nash Equilibrium), and allocate pairwise energy trading amounts without disclosing sensitive data (via novel cryptographic protocols). Specifically, we model the trading problem as a non-cooperative Stackelberg game for all the agents (i.e., buyers and sellers) to determine the optimal price, and then derive the pairwise trading amounts. Our PEM framework can privately perform all the computations among all the agents without a trusted third party. We prove the \emph{privacy}, \emph{individual rationality}, and \emph{incentive compatibility} for the PEM framework. Finally, we conduct experiments on real datasets to validate the effectiveness and efficiency of the PEM.
\end{abstract}


\section{Introduction}
\label{sec:intro}
Distributed energy resources (DERs) have been increasingly deployed in the smart grid infrastructure to supplement the power supply with renewable energy such as solar and wind. Equipped with DERs, electricity consumers (e.g., small homes with installed solar panels, hospitals and campuses with deployed microgrids) can also be considered as suppliers that have reduced their dependence on the electricity grid \cite{sellenergy13}. Recently, multi-agent systems in the smart grid \cite{mas1} have attracted significant interests by considering the smart homes or microgrids as distributed agents \cite{CerquidesPR15,XieHW19}. 
In reality, smart homes or microgrids may generate excessive energy that cannot be consumed immediately during routine operations. A current solution to deal with the excessive energy is to either consume/waste it or sell it to the main grid \cite{Poor3party15,solarray}, even if many of the smart homes/microgrids have been equipped with local storage devices. Essentially, from the economic perspective of such multi-agent systems and social welfare, transmitting excessive energy back to the main grid or storing the energy is not an ideal outcome, compared to involving more consumers (which requests external energy) to receive the excessive electricity and consume them immediately \cite{lo3energy}.

To this end, the smart grid begins to incentivize agents with local energy to cooperate with each other, e.g., decentralized power supply restoration \cite{AgrawalKV15}, energy sharing \cite{HongESharing} and three-party energy trading \cite{Poor3party15}. Inspired by them, we study the \emph{distributed energy trading} problem which enables smart homes or microgrids to sell their excessive energy to other consumers besides selling back to the power market monopoly, the main grid \cite{WtushP2P,ZHANGp2pbid}. It will greatly benefit all the agents: (1) sellers can receive more rewards with a trading price generally higher than the price requested by the main grid, (2) buyers can reduce their costs (i.e., electricity bill) with the trading price generally lower than the retail price of the main grid \cite{mckenna2013photovoltaic}, and (3) interactions/loads between the consumers and the main grid can be reduced to provide better reliability via autonomy \cite{Fioretto0PMR17}.

\begin{figure}[!h]
	\centering
		\includegraphics[angle=0, width=1\linewidth]{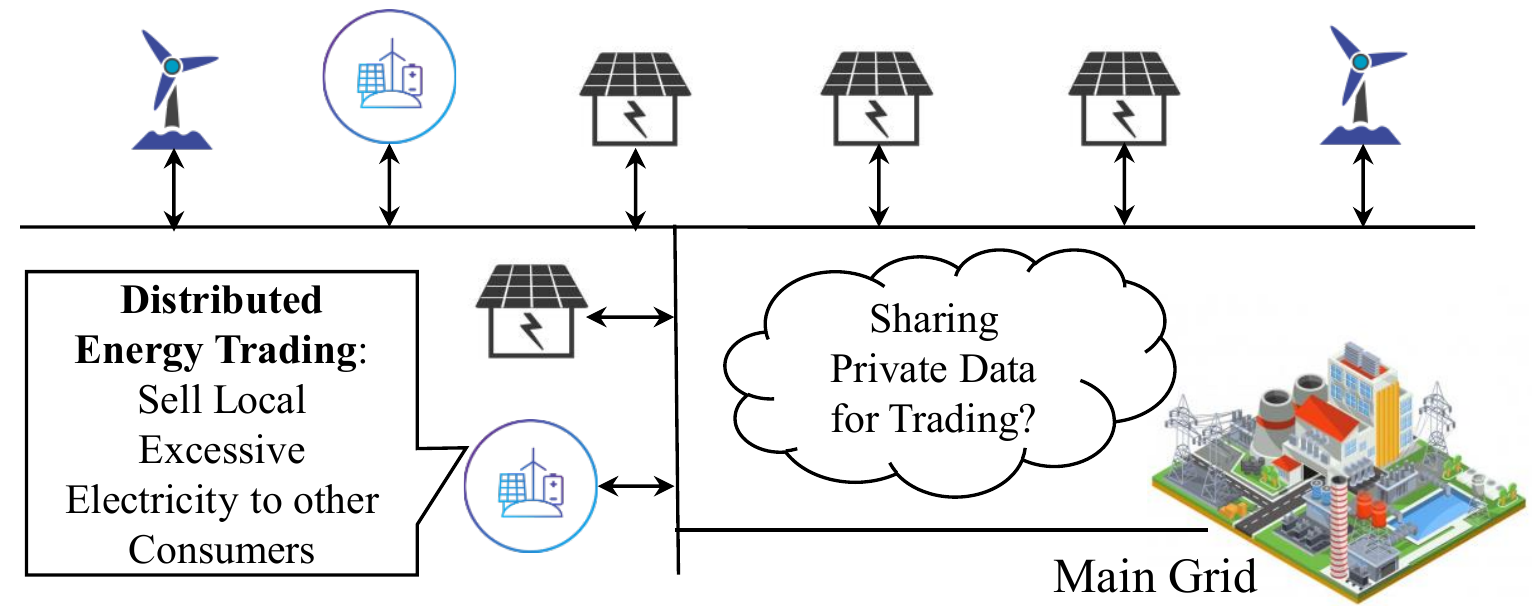}
	\caption[Optional caption for list of figures]
	{Distributed Energy Trading}
	\label{fig:p2p}
\end{figure}

However, as shown in Figure \ref{fig:p2p}, distributed energy trading requests significant amounts of local data from all the agents (e.g., each seller/buyer's local generation and demand load at different times) to compute the \emph{optimal price} and allocate the \emph{energy trading amounts} for all the sellers and buyers \cite{Poor3party15}. Disclosing such local data for computation would explicitly compromise their privacy. For instance, local generation reveals the generation capacities and time series generation patterns \cite{Kursawe:2011:PAS:2032162.2032172}, and the local demand load reveals consumption patterns (e.g., which appliance is used at which time) \cite{AcsC11,HongLW17}. 

To address such privacy concerns, we propose a novel privacy preserving distributed energy trading framework, namely ``Private Energy Market (PEM)'' in which all the agents privately compute the optimal price (ensured by a Nash Equilibrium of a designed Stackelberg game) and allocate pairwise energy trading amounts without disclosing local data. To this end, our PEM framework ensures that all the computations are performed in novel cryptographic protocols under the theory of secure multiparty computation (MPC) \cite{Yao86,Goldreich87} which provides provable privacy guarantee. Thus, the major contributions of this paper are summarized as follows: 

\begin{enumerate}

    \item To our best knowledge, the propose PEM is the first privacy preserving distributed energy trading framework, which enables all the agents on the electric grid to privately compute their optimal trading price (ensured by a Nash Equilibrium) and pairwise trading amounts, as well as complete their pairwise transactions without disclosing their private data (via cryptographic protocols). 
    \item We model a Stackelberg game \cite{stackgame} in the PEM framework, which ensures \emph{privacy} \cite{Yao86}, \emph{individual rationality}, and \emph{incentive compatibility} \cite{Nisan:2007:AGT:1296179} for all the agents. Theoretical analyses are given to prove all of the three properties.
    
    \item We implement a prototype for the proposed PEM framework with negligible latency in real time. We also conduct substantial experimental evaluations on real datasets to validate the system performance of the PEM.
    
\end{enumerate}

The remainder of this paper is organized as follows. In Section \ref{sec:pre}, we formulate the problem, discuss the threat model and the PEM framework. Section \ref{sec:mechanism} illustrates the energy trading model under the Stackelberg game. Section \ref{sec:pptrading} presents the cryptographic protocols for the PEM. Theoretical analyses are given in Section \ref{sec:analysis}. Section \ref{sec:discuss} gives some discussions. In Section \ref{sec:exp}, we elaborate our system implementation and demonstrate the experimental results. Section \ref{sec:related} reviews the literature, and Section \ref{sec:concl} draws the conclusions.

\section{Preliminaries}
\label{sec:pre}
In this section, we present some preliminaries for the distributed energy trading and the PEM framework. 

\subsection{Distributed Energy Trading}

We first introduce the background \cite{Poor3party15,CerquidesPR15, HongIJER15,mckenna2013photovoltaic} on the power grid, where agents represent the consumers with local generation, e.g., smart homes, and microgrids. 

\begin{itemize}
\item Energy trading occurs over a fixed length of periods, each of which is referred as a ``trading window''. All the agents complete their transactions (either selling or buying energy) within each trading window. 

\item Agent can be a buyer in a trading window, and a seller in another trading window, but cannot be both in any trading window (otherwise, its payoff would not be optimal \cite{HongIJER15}).

\item Each seller can decide how much energy it consumes (including charging its battery if available \cite{solarray}) and how much energy is available in the current trading window. 

\item We assume that the main grid has unlimited power supply with a higher price than distributed trading \cite{mckenna2013photovoltaic}, and energy is transmitted with a negligible loss. 
\end{itemize}

We denote the main grid as $M$ and the set of agents as $\Phi$ (with cardinality $|\Phi|$). For each agent $H_i, i\in[1,|\Phi|]$, we denote its generation (e.g., from solar panels) and demand load in trading window $t$ as $g_i^t$ and $l_i^t$, respectively. Each agent $H_i$ optionally installs an energy storage device or battery \cite{ninghuiBat12} with capacity $Cap_i$ (the maximum energy storage after charging), which can be specified as 0 (if ``no battery''). 
Denoting the energy amount charging into or discharging out of the battery as $b_i^t$ (in trading window $t$), if charging, we have $b_i^t>0$; if discharging, we have $b_i^t<0$. Then, we can define the net energy of $H_i$ as $sn_i^t$:

\begin{equation}
 	sn_i^t=g_i^t-l_i^t-b_i^t
\label{eq:nete}
\end{equation}

In every trading window $t$, each agent $H_i$ will be classified as either buyer or seller according to their net energy: (1) if $sn_i^t>0$, $H_i$ is a seller, (2) if $sn_i^t<0$, $H_i$ is a buyer, and (3) if $sn_i^t=0$, $H_i$ will be off market. Then, we formally define $\Phi_s^t=\{\forall H_i\in \Phi, sn_i^t> 0\}$ as the set of sellers and $\Phi_b^t=\{\forall H_j\in \Phi, sn_j^t< 0\}$ as the set of buyers, where the market supply of sellers $E_s^t$ and the market demand of buyers $E_b^t$ can be derived as: 
 \begin{equation}
 \small
 E_{s}^t=\sum_{H_i\in \Phi_s^t}sn_i^t>0\text{ and } E_{b}^t=\sum_{H_j\in\Phi_b^t}|sn_j^t|
 \label{eq:market}
 \end{equation}

\noindent\textbf{Optimal Trading Price}. At the end of every trading window, a seller can store the unsold energy or sell the unsold energy to the main grid \cite{solarray}. However, the price offered by the main grid (denoted as $pb_g^t$) is much lower than the regular retail electricity price for purchasing from the grid (denoted as $ps_g^t$) \cite{mckenna2013photovoltaic}. In the energy trading market, while trading energy in window $t$, all the buyers and sellers will jointly learn an optimal price $p^t$ between $pb_g^t$ and $ps_g^t$ \cite{mckenna2013photovoltaic} where all the players achieve an equilibrium in a game (with \emph{individual rationality} and \emph{incentive compatibility} \cite{microeco}).

PEM also sets an acceptable market price range $[p_l, p_h]$ to incentivize the sellers or buyers to join the trading \cite{mckenna2013photovoltaic} such that
the price $p^t$ in the trading window $t$ satisfies:

\begin{equation}
pb_g^t<p_l\leq p^t\leq p_h<ps_g^t
\label{eq:price}
\end{equation}

which is set by the PEM rather than specific agents.\footnotemark[1]\footnotetext[1]{If $p^t>ps_g^t$, all the rational buyers will purchase energy directly from the main grid; if $p^t<pb_g^t$, all the rational sellers will sell the energy directly to the main grid. Thus, PEM specifies a reasonable price range.}

Section \ref{sec:mechanism} will illustrate how to derive the optimal price and allocate energy trading amounts in every trading window.

\subsection{Threat Model}

More importantly, our PEM framework addresses the privacy concerns of all the participants (e.g.,.  agents with local energy) in the distributed energy trading. Specifically, to realize the energy trading, all the agents $\forall i\in[1,|\Phi|], H_i$ should share its local private information to a trusted third party so as to compute their optimal price as well as allocating pairwise energy trading amounts. However, such shared local information are sensitive in general \cite{SankarRMP13,HongLW17,TanGP13}, e.g., $H_i$'s local energy generation amount, energy consumption amount, battery storage amount, and its utility parameter (which are detailed in Section \ref{sec:mechanism}). 

To tackle the above concerns, we propose the PEM framework (without a trusted third party) based on efficient cryptographic protocols \cite{Yao86,Goldreich87} to privately function distributed energy trading without disclosing local information. We define the \emph{threat model} in the distributed energy trading as below:

\begin{itemize}
    
    \item We assume semi-honest adversarial model for preserving the privacy in our cryptographic protocols: all the agents are curious to learn private information from each other \cite{Yao86,Goldreichenc} but do not maliciously corrupt the protocol.
    
    \item Besides the semi-honest model, all the agents have the incentive to improve its payoff by cheating on its data.
    
    \item All the messages in the framework are assumed to be transmitted in a secure channel.
    
\end{itemize}

\subsection{PEM Framework}

To sum up, PEM will provide the following three properties against the adversaries: 

\begin{itemize}
    \item \textbf{Privacy}: each seller/buyer's privacy is protected in the PEM with provable privacy guarantee.
    \item \textbf{Individual Rationality}: each seller/buyer has a higher payoff by participating in the PEM.
    \item \textbf{Incentive Compatibility}: each seller/buyer cannot improve its payoff by untruthfully changing its strategy.
\end{itemize} 

Section \ref{sec:pptrading} will illustrate the cryptographic protocols for our PEM framework, and Section \ref{sec:analysis} will analyze the privacy/security and incentive compatibility to protect the trading under the threat model defined earlier.

\section{Distributed Energy Trading}
\label{sec:mechanism}
In this section, we first present the distributed trading scheme for PEM without privacy consideration.  
\subsection{Incentive Measurement}

We first define two functions to measure the incentives for both sellers and buyers in the trading \cite{Poor3party15}. The utility function measures the payoff received by each seller while the cost function measures how much each buyer pays. 

\vspace{0.05in}

\noindent\textbf{Seller's Utility Function} \cite{UtilityFunc10,ninghuiBat12} is defined to quantify the total utility of any seller $H_i \in \Phi_s^t$ in trading window $t$:

\begin{equation}
\small
U_i^t=k_i^t\log(1+l_i^t+\epsilon_i^t*b_i^t)+p^t*(g_i^t-l_i^t-b_i^t)
\label{eq:utility}
\end{equation}
where $k_i^t>0$ is the load behavior preference parameter of the seller $H_i$ (\emph{either locally consuming more energy or selling them}), $p^t$ is the market price. $l_i^t$ and $g_i^t$ are defined as the load and generation of $H_i$. For the battery, $b_i^t$ is defined as the energy charging/discharging amount: \emph{charging if positive (as additional load)} and \emph{discharging if negative (as additional supply)}; $\epsilon_i^t \in (0,1)$ represents the loss coefficient for the battery, which measure the ratio of battery's  contribution amount as load (charging) utility. 

\vspace{0.05in}

\noindent\textbf{Buyer's Cost Function} is defined to measure the cost of any buyer $H_j\in \Phi_b^t$ from the energy market and main grid:

\begin{equation}
\small
C_j^t=p^t*x_j^t+ps_g^t*(l_j^t+b_j^t-g_j^t-x_j^t)
\label{eq:cost}
\end{equation}

Similarly, $l_j^t$, $g_j^t$, and $b_j^t$ denote the buyer's local load, generation, and battery charging/discharging amounts, respectively.
Moreover, $x_j^t$ is defined as the energy amount that $H_j$ purchased from the trading market, thus we have $0<x_j^t\leq l_j^t+b_j^t-g_j^t$. 

\subsection{Stackelberg Game for PEM}
\label{sec:stackel}

To further pursue the cooperation of agents, two coalitions are formed based on each agent's net energy in every trading window (seller coalition and buyer coalition; the agents in two coalitions change over time). In our PEM framework, the seller coalition sells energy with the total supply while the buyer coalition purchases energy with their total demand, and their shares of energy to sell/buy are allocated proportional to their input shares (as detailed in Section \ref{sec:dist}). Such trading mechanism could make the market more stable, and guarantee the payoffs for conservative sellers/buyers who may not want to fully compete with other sellers/buyers. 
\subsubsection{Stackelberg Game}
 
Per the two (utility and cost) functions defined for sellers and buyers, the objectives of two coalitions consist of two aspects: (1) buyers incline to minimize their costs (as a coalition); (2) sellers incline to maximize their utility. To learn the optimal price, we propose a Stackelberg game for seller and buyer coalitions \cite{stackgame}.

Specifically, the market supply (from agents) is generally less than market demand (since renewable energy cannot feed all the load in current practice \cite{Poor3party15}).  Therefore, in the Stackelberg game, the buyer coalition is specified as the leader while the seller coalition is defined as the follower (otherwise, sellers will dominate the market). Then, the game $\mathcal{G}$ can be formally defined as:

\begin{equation}
\mathcal{G}=\{\Phi_b^t\cup \Phi_s^t, \{l_i^t\}_{H_i \in \Phi_s^t}, \{U_i^t\}_{H_i \in \Phi_s^t}, p^t, \Gamma^t\}    
\end{equation}
with the following components in each trading window $t$:
\begin{itemize}

\item the buyer coalition $\Phi_b^t$ is the leader to set up the price while the seller coalition $\Phi_s^t$ chooses their strategies as a response to the proposed price.

\item $\{l_i^t\}_{H_i \in \Phi_s^t}$ is the set of load profiles of all the sellers
(strategies) to maximize their payoffs.  

\item $U_i^t$ is the utility function of seller $H_i$. 

\item $p^t$ is the price proposed by the buyer coalition.

\item $\Gamma^t$ is the total cost for the buyer coalition:

\begin{equation}
\small
    \Gamma^t=\sum_{H_j \in \Phi_b^t}C_j^t=
    p^t*E_s^t+ps_g^t*(E_b^t-E_s^t)
  \label{eq:total_cost}
\end{equation}

where $E_s^t$ and $E_b^t$ are the market supply/demand (Eq. \ref{eq:market}). 

\end{itemize}

Then, the objective of the model is to minimize the total costs of the buyers/leader and to maximize the individual utility function of each seller/follower (such that the seller coalition's total utility is also maximized) by choosing their own strategies. We define the equilibrium as below:

\begin{definition}
The set of strategies $(\{l_i^{t\ast}\}_{H_i \in \Phi_s^t}, p^{t\ast})$ is an equilibrium for the game $\mathcal{G}$, if and only if it satisfies:
\begin{align*}
    U_i^t(\{l_i^{t*}\}_{H_i \in \Phi_s^t}, p^{t\ast})&\geq  U_i^t(\{l^{t}_k, \{l_i^{t*}\}_{H_i \in \Phi_s^t, i\ne k}\}, p^{t*})\\
    \Gamma^t (\{l_i^{t*}\}_{H_i \in \Phi_s^t}, p^{t\ast})&\leq  \Gamma^t (\{l_i^{t*}\}_{H_i \in \Phi_s^t}, p)
\end{align*}

\emph{where} $p_l\leq p^{t*}\leq p_h$. 
\end{definition}

Therefore, we seek for the equilibrium of this game in which the follower (aka. sellers) derives the best response to the optimal price proposed by the leader (aka. buyers). At this equilibrium, neither the leader nor any follower can increase its payoff via any \emph{unilateral strategic move}. In other words, when the game reaches the equilibrium, the buyers cannot reduce the cost by decreasing the price $p^t$ while the sellers cannot improve their utility by adjusting their strategies on load profiles $\{l_i^t\}_{H_i \in \Phi_s^t}$. 

\subsubsection{Optimal Price} 
We prove the existence and uniqueness of the equilibrium \cite{stackgame,yang2012crowdsourcing, Poor3party15} for $\mathcal{G}$:

\begin{lemma}
A unique equilibrium $(\{l_i^t\}_{H_i \in \Phi_s^t}, p^{t\ast})$ exists. 
\label{th:se}
\end{lemma}
\begin{proof}
 First, we get the second derivative of the utility function $U_i^t$ (Eq. \ref{eq:utility}):
\begin{equation}
\small
\frac{\partial^2 U_i^t}{{\partial l_i^t}^2}=\frac{-k_i^t}{(1+l_i^t+\epsilon_i^tb_i^t)^2}
\end{equation}
which is always less than 0 since $k_i^t>0$. The utility function is concave with $l_i^t$. Then given any price $p^t$, each seller $H_i\in\Phi_s^t$ can only find a unique $l_i^t$ to get its maximum utility. On the contrary, the buyers can also find the optimal price while the sellers specify their load profiles in the Nash Equilibrium. Thus, the equilibrium $(\{l_i^t\}_{H_i \in \Phi_s^t}, p^{t\ast})$ exists.

Second, to prove the uniqueness of the equilibrium, we need to prove that the optimal price is unique for the minimum cost of the buyer coalition (leader of the game). We first find the optimal load profile for each seller $H_i\in\Phi_s^t$: $l_i^t$. We then get the first derivative of $H_i$'s utility function (whose value should be $0$ for the maximum utility):
\begin{equation}
\small
\frac{\partial U_i^t}{\partial l_i^t}=\frac{k_i^t\epsilon_i^t}{(1+l_i^t+\epsilon_i^tb_i^t)}-p^t=0
\end{equation}
Thus, we get the optimal load profile for seller $H_i$:

\begin{equation}
\small
l_i^t=\frac{k_i^t\epsilon_i^t}{p^t}-1-\epsilon_i^tb_i^t
\label{eq:opt}
\end{equation}

Replacing $l_i^t$ in the total cost function $\Gamma^t$, then we get the second derivative of $\Gamma^t$:
\begin{equation}
\small
\frac{
\partial^2\Gamma^t}{{\partial p^t}^2}=\sum_{H_i\in \Phi^t_s}\frac{2ps_g^tk_i^t}{{(p^t)}^3}>0
\end{equation}

Then, $\Gamma^t$ is strictly convex with $p^t$, which generates a unique optimal price. Thus, equilibrium $(\{l_i^t\}_{H_i \in \Phi_s^t}, p^{t\ast})$ is unique. This completes the proof.
\end{proof}

To find the optimal price $p^{t\ast}$ in game $\mathcal{G}$, we calculate the first derivative of $\Gamma^t$: 

\begin{equation}
\small
      \frac{
\partial \Gamma^t}{{\partial p^t}}=\sum_{H_i\in\Phi^t_s}(g_i^t+1+\epsilon_i^tb_i^t-b_i^t)-\frac{ps_g^t\sum_{H_i\in\Phi^t_s}k_i^t}{(p^t)^2}=0
\label{eq:x}
\end{equation}
Solving Eq. \ref{eq:x}, we have
\begin{equation}
\small
    \widehat{p^t}=\sqrt{\frac{ps_g^t\sum_{H_i\in\Phi^t_s}k_i^t}{\sum_{H_i\in\Phi^t_s}(g_i^t+1+\epsilon_i^tb_i^t-b_i^t)}}
    \label{eq:optimal_p}
\end{equation}

Therefore, we can get the optimal price $p^{t\ast}$ 
by integrating Eq. \ref{eq:optimal_p} and \ref{eq:price}.

\begin{equation}
\small
p^{t\ast}=\begin{cases}
	\widehat{p^t}, \qquad \widehat{p^t}\in[p_l, p_h]\\
 	p_l,\qquad \widehat{p^t}<p_l\\
    p_h,\qquad \widehat{p^t}>p_h \\
 	\end{cases}\label{eq:optimal}
\end{equation}

Replacing $p^t$ in the load profile $l_i^t$ (Eq. \ref{eq:opt}) with $p^{t\ast}$, we can get the optimal load profile (strategy) for each seller $H_i$: 
\begin{equation}
\small
l_i^{t\ast}=\frac{k_i^t\epsilon_i^t}{p^{t\ast}}-1-\epsilon_i^tb_i^t
\end{equation} 

Note that if there is no battery installed for the seller, we thus have $b_i^t=0$.

\subsection{Trading Scheme in an Extreme Market}

If the market supply in the PEM is greater than or equal to the market demand (this rarely occurs in the current smart grid infrastructure, ``extreme market''), to maintain a robust market, the market electricity price should be set to the lower bound $p_l$ which is still greater than the price $pb_g^t$ offered by the main grid. Different from the general market case, the sellers also maximize their utilities by selling the remaining energy to the main grid and the buyer coalition will buy the electricity for all its demand from the market (to minimize their costs). 

\subsection{Energy Distribution and Payment}
\label{sec:dist}

Considering $E_s^t<E_b^t$ as the general market and $E_s^t\geq E_b^t$ as the extreme market, our PEM framework allocates trading amount for each pair of buyer and seller based on the demand (general market) or supply ratio (extreme market) out of the total market supply and demand to ensure fairness of distribution. We now discuss the allocation strategies for the two markets.

\vspace{0.05in}

\noindent\textbf{General Market:} the optimal price $p^{t\ast}$ is proposed by the buyer coalition in the Stackelberg Game and all the market supply should be sold to the buyer coalition with price $p^{t\ast}$. In the buyer coalition, the amount of electricity should be allocated in terms of their demand ratio out of the total demand $E_b^t$. Then, each buyer $H_j\in \Phi_b^t$ requests energy with the amount $e_{ij}=sn_i^t*\frac{|sn_j^t|}{E_b^t}$ from seller $H_i\in\Phi_s^t$, and pays $m_{ji}=p^{t\ast}e_{ij}$ to seller $H_i$. 
\vspace{0.05in}

\noindent\textbf{Extreme Market:} the price is directly set as $p_l$. Similarly, each seller $H_i\in\Phi_s^t$ sells the amount of $e_{ij}=|sn_j^t|*\frac{sn_i^t}{E_s^t}$ to buyer $H_j\in\Phi_b^t$ and receives the payment of $m_{ji}=p_le_{ij}$ from buyer $H_j$.

\section{Cryptographic Protocols}
\label{sec:pptrading}
In this section, we present the cryptographic protocols for the distributed energy trading in PEM. 

\subsection{Cryptographic Building Blocks}

We adopt homomorphic encryption \cite{Paillier99} and garbled circuit \cite{Yao86,Goldreich87} as the building blocks to construct our protocols.

\vspace{0.05in}

\noindent\textbf{Homomorphic Encryption} (e.g., Paillier cryptosystem \cite{Paillier99}) is a semantically-secure public key encryption to generate the ciphertext of an arithmetic operation between two plaintexts by other operations between their ciphertexts. It has the additional property that given any two encrypted messages $E(A)$ and $E(B)$, we have $E(A+B)=E(A)*E(B) $, where $*$ denotes the multiplication of ciphertexts (in some abelian groups).

\vspace{0.05in}

\noindent\textbf{Garbled Circuit} was originally proposed by Yao \cite{Yao86}. It enables two parties to jointly compute a function without disclosing their private inputs where one party creates the garbled circuit and the other party evaluates the circuit to derive the result of the secure computation. Our protocols only incorporate garbled circuit (e.g., the \textsc{Fairplay} system \cite{MalkhiNPS04}) for realizing some light-weight computations (e.g., secure comparison) instead of the entire trading scheme.

\subsection{Overview of the PEM}

\begin{figure}[!h]
	\centering
		\includegraphics[angle=0, width=1\linewidth]{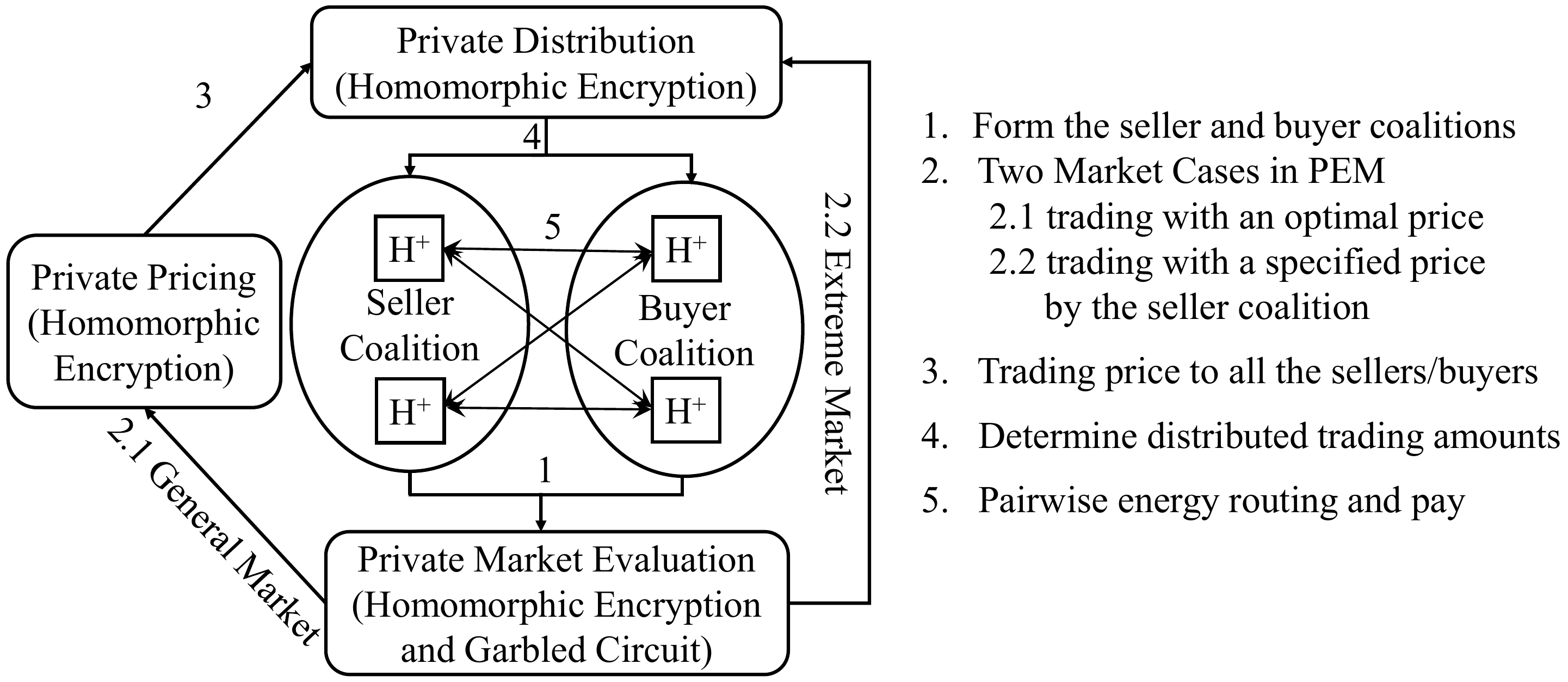}
	\caption[Optional caption for list of figures]
	{Overview of the PEM Framework}
	\label{fig:overview}
\end{figure}

As shown in Figure \ref{fig:overview} and Protocol \ref{algm:pem}, in the PEM framework, all the agents firstly form the seller and buyer coalitions (Initialization). Then, in Private Market Evaluation, the two coalitions securely evaluate the market. If a general market case is returned, Private Pricing is executed to securely compute the optimal price. For both general and extreme market cases, Private Distribution is executed to complete the trading. 

\IncMargin{1em}
 \begin{algorithm}[!ht]
 \small
\For{agent $H_i \in \Phi$}
{$H_i$ generates key pair $(pk_i, sk_i)$, and shares $pk_i$ in $\Phi$
}
\For{each trading window t}
{
 Initialize seller and buyer coalitions: $\Phi_s^t$, $\Phi_b^t$

$\Phi_s^t$ and $\Phi_b^t$ execute \textbf{Private Market Evaluation}

\If{$E_s^t< E_b^t$ (general market)}
{Execute \textbf{Private Pricing} (Protocol \ref{algm:SecPr}):
$p^t=p^{t\ast}$
}

\Else{

Set the current price $p^t=p_l$ (extreme market)}

$\Phi_s^t$ and $\Phi_b^t$ execute \textbf{Private Distribution} (Protocol \ref{algm:SecDis})

}

\caption{Private Energy Market (PEM)}\label{algm:pem}
\end{algorithm}

\DecMargin{1em}

\subsection{Initialization} \label{sec:ini}

Since secure computation in the PEM primarily utilizes the Homomorphic encryption (e.g., Paillier Cryptosystem \cite{Paillier99}), each seller/buyer locally generates its own public-private key pair and shares all their public keys. At the beginning of each trading window, each agent claims its role as buyer or seller or off the market to form the seller and buyer coalitions. If the seller coalition is empty ($E_s^t=0$), all the buyers should buy energy from the main grid with the retail electricity price.

\IncMargin{1em}

\begin{algorithm}[!ht]
\small

Randomly choose $H_{r1} \in \Phi_s^t$ with key pair $(pk_{r1}, sk_{r1})$

\For{each $H_j\in \Phi_b^t$}
{
$H_j$ generates random nonce $r_j$ and initializes $C=1$

$H_j$ computes $C \leftarrow C*Enc_{pk_{r1}}(|sn_j^t|+r_j)$
}

The last agent in $\Phi_b^t$ sends $C$ to $\Phi_s^t\setminus H_{r1}$ 

\For{each $H_i \in \Phi_s^t$}{
$H_i$ generates random nonce $r_i$

$H_i$ computes $C \leftarrow C*Enc_{pk_{r1}}(r_i)$
}

The last agent in $\Phi_s^t$ sends $C$ to $H_{r1}$

$H_{r1}$ obtains $R_b=\sum_{H_j\in \Phi_b^t}(|sn_j^t|+r_j)+\sum_{H_i\in \Phi_s^t}r_i$ by decrypting the ciphertext with $sk_{r1}$

Randomly choose $H_{r2} \in \Phi_b^t$ with key pair $(pk_{r2}, sk_{r2})$

\textbf{Repeat} Lines 2-10 with $H_{r1} \leftarrow H_{r2}$

$H_{r2}$ obtains $R_s=\sum_{H_i\in\Phi_s^t}(sn_i^t+r_i)+\sum_{H_j\in \Phi_b^t}r_j$ by decrypting the ciphertext with $sk_{r2}$

$H_{r1}, H_{r2}$ execute secure comparison with input $R_b, R_s$

\If{$R_s < R_b$}{
    \KwRet general market}
\Else{
    \KwRet extreme market}

	\caption{Private Market Evaluation}\label{algm:SecCoop}
\end{algorithm}

\DecMargin{1em}

\subsection{Private Market Evaluation}\label{sec:seccoop}
After the initialization, PEM determines the market to be a general or extreme market, where the seller coalition $\Phi_s^t$ and buyer coalition $\Phi_b^t$ jointly aggregate their private net energy, and then compare the overall supply $E_s^t$ and demand $E_b^t$. Specifically, there are ``two rounds'' of aggregations. In the first round, an arbitrary seller $H_{r1}$ will be chosen, then each buyer $H_j\in \Phi_b^t$ encrypts the demand $|sn_j^t|$ plus a random nonce $r_j$ using $H_{r1}$'s public key $pk_{r1}$, i.e., $\small Enc_{pk_{r1}}(|sn_j^t|+r_j)$ for summing up $ \small \sum_{H_j\in \Phi_b^t}Enc_{pk_{r1}}(|sn_j^t|+r_j)$. The ciphertext will be sent to one random seller in the seller coalition except $H_{r1}$. Similarly, each seller $H_i$ of the seller coalition generates a nonce $r_i$ and encrypts it for summing up the value in the ciphertext. Finally, $H_{r1}$ decrypts the ciphertext to get the aggregated value $R_b=\sum_{H_j\in \Phi_b^t}(|sn_j|+r_j)+\sum_{H_i\in \Phi_s^t}r_i$ with its private key (see Lines 2-10 in Protocol \ref{algm:SecCoop}). The second round is similar to the first: one random selected buyer $H_{r2}$'s public key $pk_{r2}$ is used to aggregate the $sn_i+r_i$ of each seller $H_i\in \Phi_s^t$ and the nonce $r_j$ of each buyer to get $R_s=\sum_{H_i\in \Phi_s^t}(sn_i+r_i)+\sum_{H_j\in \Phi_b^t}r_j$ (see Lines 11-13).

As shown in Protocol \ref{algm:SecCoop}, neither the chosen seller $H_{r1}$ nor buyer $H_{r2}$ knows the value of $E_b^t$ or $E_s^t$ and they only obtain the aggregated random value ($R_b$ or $R_s$). Furthermore, Private Market Evaluation securely compares $R_b$ and $R_s$ by $H_{r1}$ and $H_{r2}$ to determine the market case using the garbled circuits (e.g., the \textsc{Fairplay} system \cite{fairplaymp}, see Lines 14-18). Note that the comparison result of $R_b$ and $R_s$ is equivalent to the comparison result of $E_b^t$ or $E_s^t$ since the same sum of random nonces are added to $E_b^t$ and $E_s^t$ to obtain $R_b$ and $R_s$.

\subsection{Private Pricing}
\label{sec:secpri}

If general market is returned in Protocol \ref{algm:SecCoop}, Private Pricing will be executed to find the optimal price $p^{t*}$ of the Stackelberg Equilibrium. Specifically, a seller $H_b$ will be chosen at random to securely aggregate two local values of each seller $H_i \in \Phi_s^t$: (1) $k_i^t$, and (2) $g_i^t+1+\epsilon_i^tb_i^t-b_i^t$ (locally computed). Then $H_b$ derives the optimal price $p^{t\ast}$ per the Eq. \ref{eq:optimal} in Section \ref{sec:stackel} once getting $\widehat{p^t}$, and broadcasts the optimal price $p^{t\ast}$ (see Lines 8-9 in Protocol \ref{algm:SecPr}).

\IncMargin{1em}

\begin{algorithm}[!ht]
\small

Choose randomly $H_{b} \in \Phi_b^t$ with key pair $(pk_{b}, sk_{b})$

\For{each $H_i\in\Phi_s^t$}{ $\prod^i_{s=1}Enc_{pk_b}(k_s^t) \leftarrow \prod^{i-1}_{s=1}Enc_{pk_b}(k_i^t)*Enc_{pk_{b}}(k_i^t)$}

The last agent sends $\prod^{|\Phi_s^t|}_{i=1}Enc_{pk_b}(k_i^t)$ to $H_{b}$

$H_{b}$ decrypts $\prod^{|\Phi_s^t|}_{i=1}Enc_{pk_b}(k_i^t)$ using $sk_b$ for $\sum_{H_i\in \Phi_s^t}k_i^t$

\textbf{Repeat} Lines 2-5 with $k_i^t\leftarrow g_i^t+1+\epsilon_i^tb_i^t-b_i^t$ 

$H_{b}$ decrypts the ciphertext to obtain $\sum_{H_i\in \Phi_s^t}(g_i^t+1+\epsilon_i^tb_i^t-b_i^t)$

$H_b$ calculates:  $\widehat{p^t}=\sqrt{\frac{ps_g^t\sum_{H_i\in \Phi_s^t}k_i^t}{\sum_{H_i\in \Phi_s^t}(g_i^t+1+\epsilon_i^tb_i^t-b_i^t)}}$

$H_b$ derives $p^{t\ast}$ per the Eq. \ref{eq:optimal} and broadcasts it
\caption{Private Pricing}\label{algm:SecPr}
\end{algorithm}

\DecMargin{1em}

\subsection{Private Distribution}\label{subsec:distri}

As discussed in Section \ref{sec:dist}, the trading amount of electricity between each pair of seller and buyer should be allocated in proportion to its demand/supply ratio out of the market demand/supply in both general and extreme market case. W.l.o.g., we discuss the protocol for the general market (which can be simply extended for the extreme market). For buyer $H_j$, the allocated amount of electricity from seller $H_i$ should be $e_{ij}=\frac{|sn_j^t|}{E_b^t}*sn_i^t$. Since any buyer may intend to cheat by using a larger demand $sn_i^t$ to increase its share in the allocation (reduce costs with a lower price to buy energy), the market demand \emph{cannot be directly disclosed to the buyers}. The seller coalition cannot get the market demand considering the \emph{privacy} and \emph{fairness}. Since the homomorphic encryption schemes (e.g., Paillier Cryptosystem \cite{Paillier99}) only obtain additive and/or multiplicative property (not fully homomorphic \cite{Gentry09} to securely compute ``division''), we cannot directly adopt homomorphic encryption for privately computing the pairwise allocated amounts using their input ratios.

To address such issue, we transform the ciphertext computation for the ``division/ratio'' $\frac{|sn_j^t|}{E_b^t}$ in $e_{ij}$. Specifically, each buyer $H_j$ locally computes $Enc_{pk_s}(\frac{E_b^t}{|sn_j^t|})$ with $\frac{1}{|sn_j^t|}$. Note that $\frac{1}{|sn_j^t|}$ should be multiplied by an integer $k$ to be converted to an integer. Then $Enc_{pk_s}(\frac{E_b^t}{|sn_j^t|})$ and $k$ will be sent to the seller $H_s$, and $H_s$ decrypts it to get the allocation ratio via $(\frac{E_b^t}{|sn_j^t|})^{-1}=\frac{|sn_j^t|}{E_b^t}$. The only information that the seller $H_s$ knows is the allocation ratio for buyer coalition (while $E_b^t$ and $|sn_j^t|$ are unknown). Thus, the seller $H_s$ can broadcast the allocation ratio in the seller coalition. Finally, each seller $H_i$ calculates the allocated amount of energy $e_{ij}$ and routes the energy to each buyer $H_j$; $H_j$ pays $m_{ji}$ to $H_i$ (see Lines 10-12 in Protocol \ref{algm:SecDis}). Similarly, in an extreme market, the protocol can be implemented by swapping their roles. Figure \ref{fig:pridis} illustrates the major procedures of Private Distribution (note that the random seller $H_s$ is chosen as $H_1^{+}$).

\begin{figure}[!h]
 	\centering
 		\includegraphics[angle=0, width=1\linewidth]{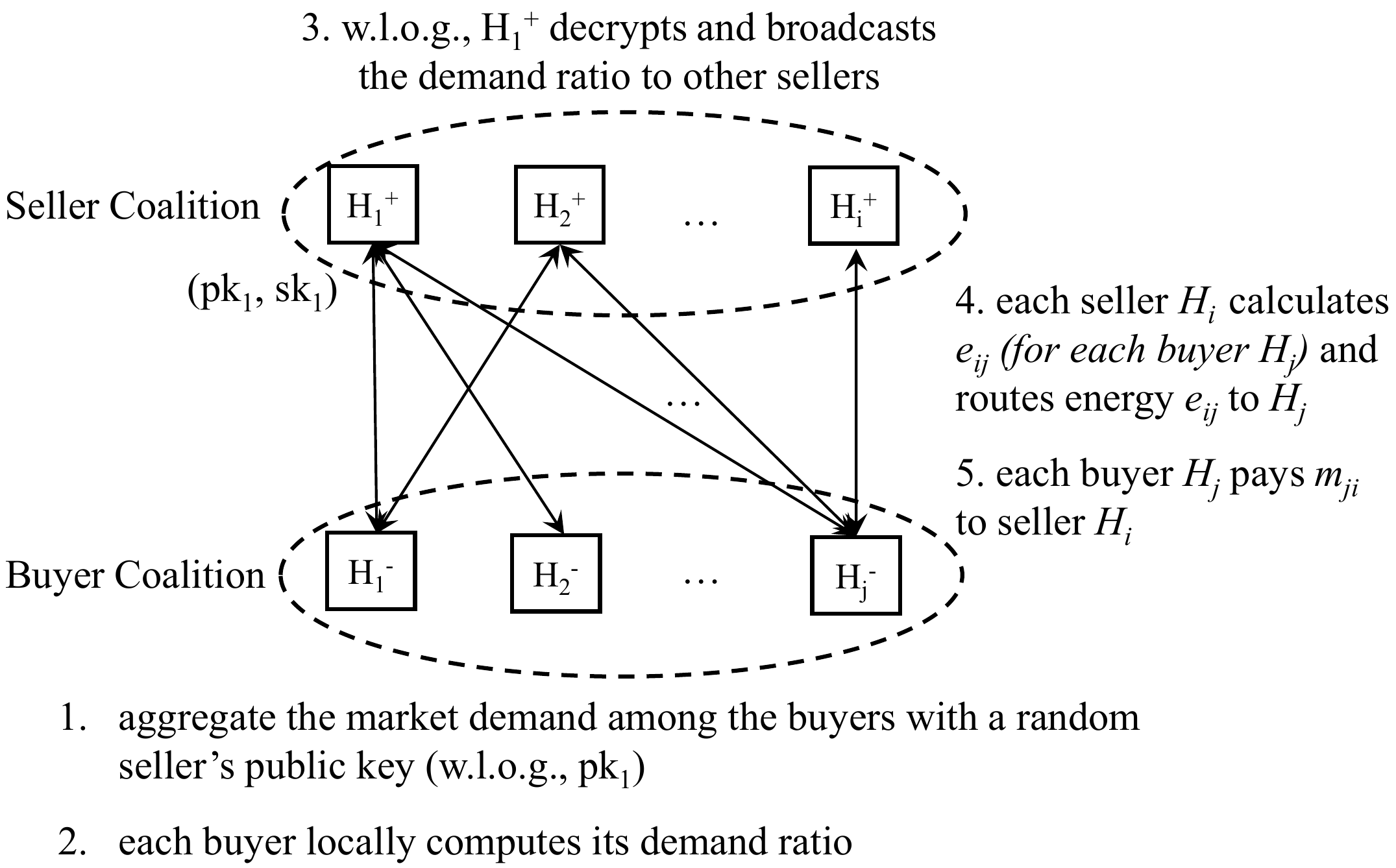}
 	\caption[Optional caption for list of figures]
 	{Private Distribution for General Market (which can be adapted for extreme market by swapping the roles of two coalitions: each buyer $H_j$ calculates $e_{ij}$ and pays $m_{ji}$)}
 	\label{fig:pridis}
\end{figure}

\IncMargin{1em}

\begin{algorithm}[!ht]
\small

\If{general market}{

Randomly choose $H_s \in \Phi_s^t$ with key pair $(pk_s, sk_s)$

\For{each $H_j\in \Phi_b^t$}{
$H_j$ computes $\prod^j_{s=1}Enc_{pk_s}(|sn_s^t|) \leftarrow \prod^{j-1}_{s=1}Enc_{pk_s}(|sn_j^t|)*Enc_{pk_s}(|sn_j^t|)$
}
The last agent broadcasts $\prod^{|\Phi_b^t|}_{j=1}Enc_{pk_s}(|sn_j^t|)$ in $\Phi_b^t$

\For{each $H_j\in \Phi_b^t$}{
$H_j$ computes and sends $\prod^{|\Phi_b^t|}_{j=1}Enc_{pk_s}(|sn_j^t|)^{\frac{1}{|sn_j^t|}}$ to $H_s$
}
$H_s$ decrypts the ciphertexts and broadcasts the allocation ratio within the seller coalition $\Phi_s^t$ 

\Repeat{each $H_i \in\Phi_s^t$ finishes transaction}{
$H_i$ computes $e_{ij}= \frac{|sn_j^t|}{E_b^t}*sn_i^t$

$H_i$ routes $e_{ij}$ to $H_j$ 

$H_j$ pays $m_{ji}=e_{ij}*p^t$ to $H_i$

}

}

\Else{
\textbf{Repeat} Lines 2-13 by replacing $\Phi_s^t$ with $\Phi_b^t$

\KwRet $e_{ij}= \frac{sn_i^t}{E_s^t}*|sn_j^t|$ and $m_{ji}= e_{ij}*p^t$
}
    \caption{Private Distribution}\label{algm:SecDis}
\end{algorithm}

\DecMargin{1em}

\section{Analysis}
\label{sec:analysis}

In this section, we give theoretical analysis for privacy, incentives, and the complexity in our PEM framework.

\subsection{Security/Privacy Analysis}
We now prove the security/privacy for the protocols in our PEM framework under the theory of secure multiparty computation \cite{Yao86,Goldreich87}, which requires each party to simulate all its received messages with only its input and output in polynomial time (``\emph{Computational Indistinguishability}'') \cite{Goldreichenc}. The PEM framework executes Private Market Evaluation, Private Distribution and possibly Private Pricing in each trading window. Then, we first examine the security of the three protocols and then discuss the composition \cite{Goldreichenc}.
\begin{lemma} 
The Private Market Evaluation (Protocol \ref{algm:SecCoop}) does not reveal any private information. 
\end{lemma}

\begin{proof}
Three different types of parties are involved in Protocol \ref{algm:SecCoop}: a randomly selected seller $H_{r1}$, a randomly selected buyer $H_{r2}$, and the remaining sellers/buyers.

We first examine the received messages of the remaining sellers/buyers. Each of them only receives a ciphertext of a random number (which cannot be decrypted without the private key), which can be polynomially simulated by repeating the encryption with the public key. Thus, the protocol does not reveal private information to them.

$H_{r1}$ and $H_{r2}$ can decrypt the ciphertexts to obtain two different random numbers $R_b$ and $R_s$ (which are decrypted with the private keys), respectively. Each random number can be polynomially simulated by generating a random number from the uniform probability distribution over range $\mathcal{F}$. Notice that the random numbers are scaled to fixed precision over a closed field (after decryption), enabling such a selection. Thus, $Pr[\sum_{i=1}^nR_b~is~simulated] = Pr[\sum_{i=1}^nR_s~is~simulated]=\frac{1}{\mathcal{F}}$. 
Finally, $H_{r1}$ and $H_{r2}$ also securely execute Fairplay to compare two random numbers for market evaluation, which does not reveal any private information (as proven in \cite{MalkhiNPS04}). 
\end{proof}

\begin{lemma} 
The Private Pricing (Protocol \ref{algm:SecPr}) only reveals non-private information $\small {\sum_{H_i\in \Phi_s^t}k_i^t}$ and $\sum_{H_i\in \Phi_s^t}(g_i^t+1+\epsilon_i^tb_i^t-b_i^t)$ to a randomly selected buyer $H_b$.

\label{lem:secpri}
\end{lemma}

\begin{proof}
This Protocol involves two different types of parties: a randomly selected buyer $H_b$ and all the sellers. 

We first analyze $H_b$'s received messages. $H_b$ can decrypt the received ciphertexts with its private key to obtain $\sum_{H_i\in \Phi_s^t}k_i^t$ and $\sum_{H_i\in \Phi_s^t}(g_i^t+1+\epsilon_i^tb_i^t-b_i^t)$. Although such two aggregated values are revealed to $H_b$, $H_b$ cannot learn any seller's private data, e.g., $k_i^t, g_i^t, b_i^t$, $\epsilon_i^t$ from the aggregated results. 

On the other hand, all the sellers receive only two ciphertexts and cannot decrypt them without the private key. Since each seller can polynomially simulate its received two ciphertexts using the public key (by repeating the encryption), we can claim that the protocol does not reveal any information to the sellers.
\end{proof}

\begin{lemma} 
The Private Distribution (Protocol \ref{algm:SecDis}) only reveals the non-private market demand ratios $\frac{E_b^t}{|sn^t_j|}, H_j\in \Phi_b^t$ to the seller coalition (in the general market), or the non-private market supply ratios $\frac{E_s^t}{|sn^t_i|}, H_i\in \Phi_s^t$ to the buyer coalition (in the extreme market).
\label{lemma:pd}

\end{lemma}

\begin{proof}
Similar to the proof in Lemma \ref{lem:secpri}, we can prove that the seller coalition can only receive the demand ratios $\frac{E_b^t}{|sn^t_j|}, H_j\in \Phi_b^t$ (from the buyer coalition) in general market. Moreover, buyer coalition can only receive the supply ratios $\frac{E_s^t}{|sn^t_i|}, H_i\in \Phi_s^t$ (from the seller coalition) in extreme market. However, they cannot learn any supply or demand from the ratios in these two cases. 
\end{proof}

\begin{theorem} 
The PEM framework only reveals the aggregated information articulated in Lemma \ref{lem:secpri} and \ref{lemma:pd}.
\end{theorem}
\begin{proof}
Since Private Market Evaluation does not reveal any privacy where the secure comparison result (either 0 or 1) can be polynomially simulated, PEM only reveals some trivial information articulated in Lemma \ref{lem:secpri} and \ref{lemma:pd} per the composition theory of secure multiparty computation \cite{Goldreichenc}. 
\end{proof}

\subsection{Incentive Analysis}

\begin{theorem} 
The PEM framework ensures individual rationality and incentive compatibility.
\label{th:incentive}
\end{theorem}

\begin{proof}
 
We first evaluate the individual rationality. In the general market, if each buyer $H_j$ directly purchases energy from the main grid at the price $ps_g^t$, which is greater than $p^{t*} \in[p_l, p_h]$, the cost will increase; if each seller $H_i$ directly sells the energy to the grid at the price $pb_g^t$, which is less than the $p^{t*}$: the payoff will decrease. In the extreme market, the buyer can buy the energy from the PEM with a lower price ($p_l<ps_g^t$) and the seller can still sell the energy with a higher price ($p_l>pb_g^t$), both of which receive more payoffs. This proves the individual rationality. 

Second, we discuss the incentive compatibility for two different markets: for the general market, we assume that there exists one seller $H_i\in \Phi_s^t$ which untruthfully utilizes its net energy ${sn_i^t}'$ by adjusting its load profile to ${l_i^t}'$. Per Lemma \ref{th:se}, there exists only one load profile $l_i^{t\ast}$ to reach the equilibrium and return the optimal price $p^{t*}$. Then, it is impossible to find another ${l_i^t}' \neq l_i^{t\ast}$ since $p^{t*}$ is derived only if all the sellers hold the $l_i^{t\ast}, H_i\in \Phi_s^t$ profile. On the contrary, as all the sellers hold the optimal load profile, the buyers cannot reduce the total costs by decreasing market price.

In addition, for the extreme market, the buyers purchase all the energy from the PEM with a lower price $p_l<ps_g^t$, then rational buyers cannot gain more payoff with untruthful participation (since the payment cannot be lower). For any rational seller $H_i$, if $H_i$ untruthfully utilizes a higher supply to increase its allocated amount of sold energy, the market price would be reduced (no additional payoff, either). This proves the incentive compatibility. 
\end{proof}

\subsection{Complexity Analysis}
\begin{lemma}
The complexity of protocols in the PEM is $O(n^2)$.
\end{lemma}

\begin{proof}
It is straightforward to analyze the complexity of algorithms in our PEM framework. First, Private Market Evaluation algorithm has complexity $O(n)$ -- securely aggregating random values is $O(n)$ while secure comparison is $O(1)$. Similarly, Private Pricing algorithm has complexity $O(n)$, and Private Distribution algorithm has complexity $O(n^2)$. Therefore, the complexity of the PEM framework is $O(n^2)$. 
\end{proof}

\section{Discussion}
\label{sec:discuss}
\noindent\textbf{Generalization of PEM}. PEM can be extended to Vehicle-to-Grid (V2G) applications \cite{al2011coordinating} by considering electrical vehicles as agents with local energy. Last but not least, the proposed PEM is a general framework for privacy preserving energy trading (focusing on privacy and incentive compatibility), which can be readily extended for ensuring privacy and incentive compatibility for other applications on the power grid (e.g., energy trading w.r.t. future prices, energy trading by possibly storing energy for the future, and demand response \cite{wang2013active}). Finally, PEM can also be adapted for trading other products, such as the allocation of spectrum in the cognitive radio networks \cite{LeiToN}, and the Wifi \& LTE sharing \cite{XuLTE}. 
     \vspace{0.03in}
     
    \noindent \textbf{Seller/Buyer Coalitions}. We forms coalitions for sellers and buyers in our PEM. First, the formation of coalition can enable the agents to cooperate to achieve more benefits/social welfare compared with trading directly with the monopoly, the main grid. Recall that coalitions make the market more stable for such emerging applications, e.g., ensuring the fairness among the seller/buyer coalition by allocating the amounts based on sellers/buyers' shares in the market supply/demand. Such setting would be more applicable for conservative agents. Nevertheless, it is also worth exploring the privacy preserving schemes for non-cooperative energy trading or fully competitive energy market \cite{wang2014game}, which left for future work. 
    \vspace{0.03in}
    
    \noindent \textbf{Malicious Model}. PEM is based on the semi-honest model, and each agent (rational) is also assumed to have incentives to cheat for payoffs. Our model can also be extended to defend against malicious agents, which may deviate the protocol (regardless of their payoffs) by faking the trading data, colluding with other agents, and/or performing advanced attacks. For instance, we can design verifiable and collusion-resistant schemes (e.g., detect the violation of data integrity, and prevent collusion by randomly picking agents \cite{Xie_TIFS19}). 
    \vspace{0.03in}

    \noindent \textbf{Scalability}. With the advancement of distributed computation \cite{Xie_TIFS19, MaharjanZZGB13}, secure computation \cite{Yao86,Goldreich87} can be applied to perform complex computation on the smart grid. Each distributed agent (e.g., a smart home) can also locally compute the data, such that the computational load of whole system can be greatly reduced. As shown in the experimental settings in Section \ref{subsec:expset}, we take advantage of the container technology, e.g., Docker, to emulate local computing agents for different smart homes in the PEM. High efficiency and scalability of PEM have been demonstrated.
    
    \vspace{0.03in}

\noindent\textbf{Blockchain Deployment}. PEM can also be integrated with the emerging blockchain technology \cite{nakamoto2008bitcoin}. Specifically, the final distribution and transaction between the sellers and buyers can be realized by the smart contract of the blockchain to ensure the integrity and truthfulness (extra anonymity and privacy should be ensured on the blockchain) \cite{kosba2016hawk}. Moreover, the on-line blockchain can also facilitate the communication of the MPC protocols in the PEM. 

\vspace{0.03in}
    
\noindent \textbf{Secure Computation}. The recent protocols/systems on secure computation (e.g., MPC-as-a-service \cite{MPCasService}, against both semi-honest and malicious adversaries \cite{furukawa2017high}, MPC for small number parties \cite{Byali:2018:FSC:3243734.3243784}) cannot be adapted to solve our problem for the following two major reasons: (1) whether the system can function real time transactions has not been validated in most of such systems (we have validated the feasibility and scalability of deploying PEM in real time in Section \ref{sec:exp}), and (2) incentive problems are not studied in most of such systems. Thus, the proposed cryptographic protocols in PEM can also complement the literature of secure computation for privacy preserving trading (which is limited to our best knowledge).

\section{Evaluation}
\label{sec:exp}

In this section, we illustrate our system implementation for the PEM framework and demonstrate the experimental results.

\subsection{Experimental Setup}\label{subsec:expset}

Our PEM framework is deployed on the NSF CloudLab platform (https://docs.cloudlab.com), of which the server has eight 64-bit ARMv8 cores with 2.4 GHZ, 64GB memory and 120GB of flash storage with Ubuntu:16.04 OS. Docker (https://docs.docker.com/) is utilized to start a container for each buyer/seller. We created the image for container based on the raw image of Ubuntu 16.04 by integrating all the system environments (e.g., JRE and JDK), and source codes. 

We conducted the experiments on 300 smart homes' real power generation data (solar panels) and load data over one day (available at UMASS Trace Repository \cite{Barker2012}). We tune the following parameters in evaluations:

\begin{itemize}
  
\item the number of smart homes $n\in [100,300]$; 
\item the number of trading windows $m\in[1, 720]$: from 7:00AM to 7:00PM (a trading window per minute); 
\item the key size: 512/1024/2048-bit.
   
\end{itemize}

\noindent\textbf{Benchmark.} There is no existing schemes which can be directly applicable to our problem setting. Then, we use the traditional energy trading (without PEM) as the benchmark: all the agents directly purchase energy from the main grid. Specifically, if a seller (with excessive energy) will sell them back to the main grid with the offered price $pb_g^t$, and the buyer (short of energy) will buy energy from the main grid with the retail electricity price $ps_g^t$. We set the retail price as $ps_g^t$=120 cents/kWh and offered price from the main grid $pb_g^t$= 80 cents/kWh. We also set the price range of PEM as $[90, 110]$ cents/kWh. Note the interaction between agents and the main grid will increase greatly for trading without PEM.

\begin{figure}[!h]
	\centering
		\includegraphics[angle=0, width=1\linewidth]{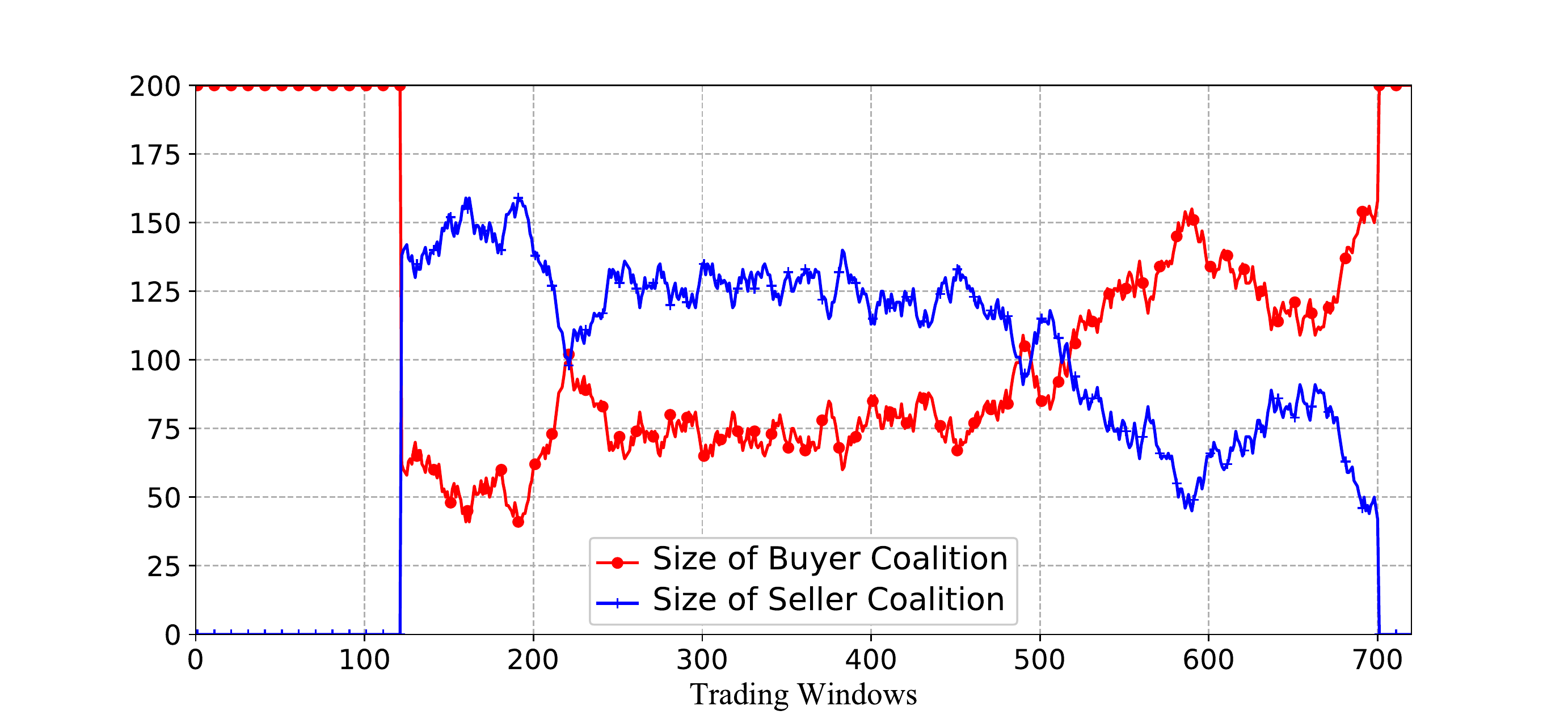}
	\caption[Optional caption for list of figures]
	{Coalition Sizes vs. Trading Windows} 
\label{fig:size}
\end{figure}

\begin{figure*}[!tbh]
	\centering
	\subfigure[Runtime (2048-bit)]{
		\includegraphics[angle=0, width=0.32\linewidth]{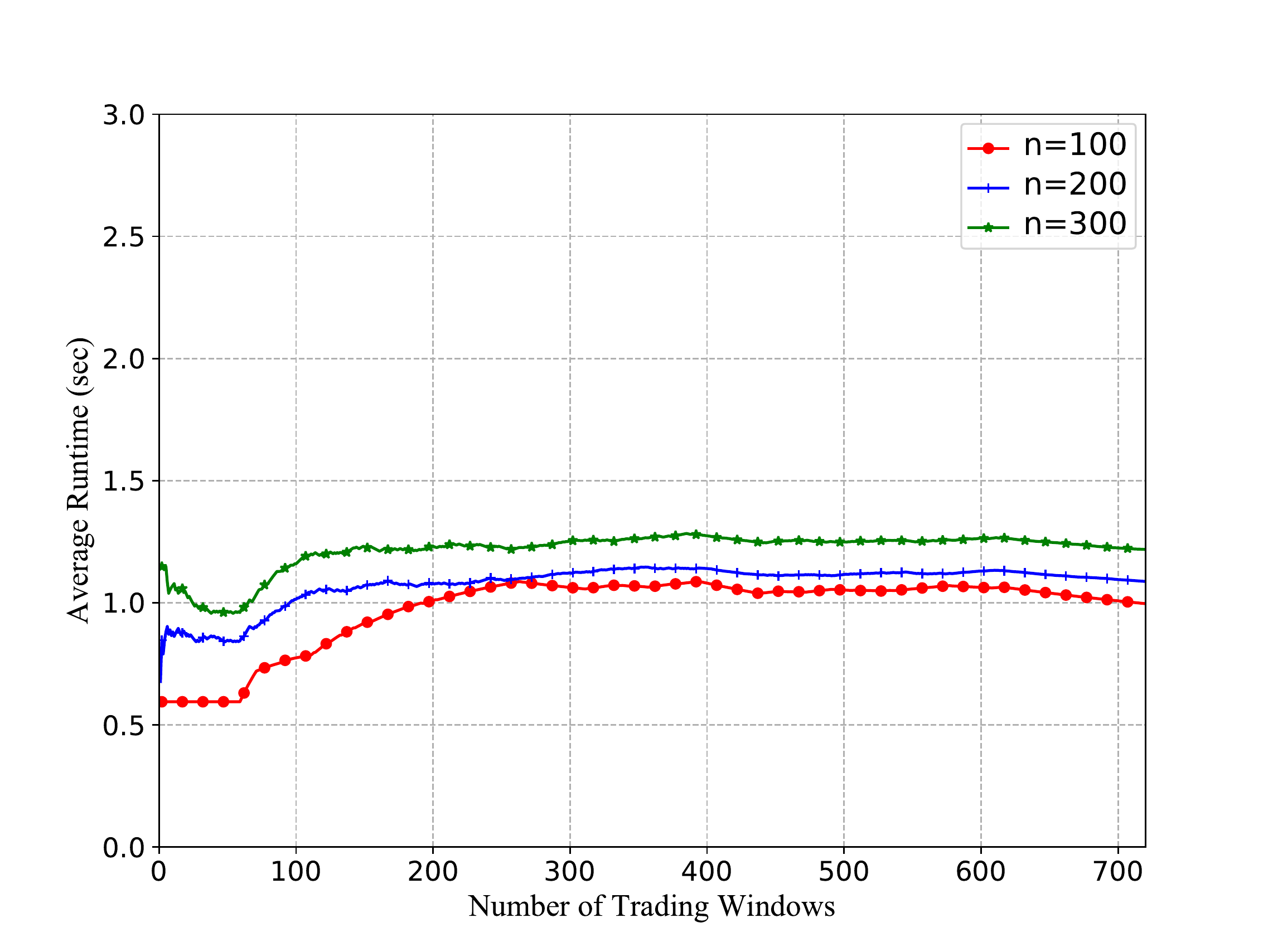}\hspace{-0.2in}
	\label{fig:1-1}}
	\subfigure[Runtime]{
		\includegraphics[angle=0, width=0.32\linewidth]{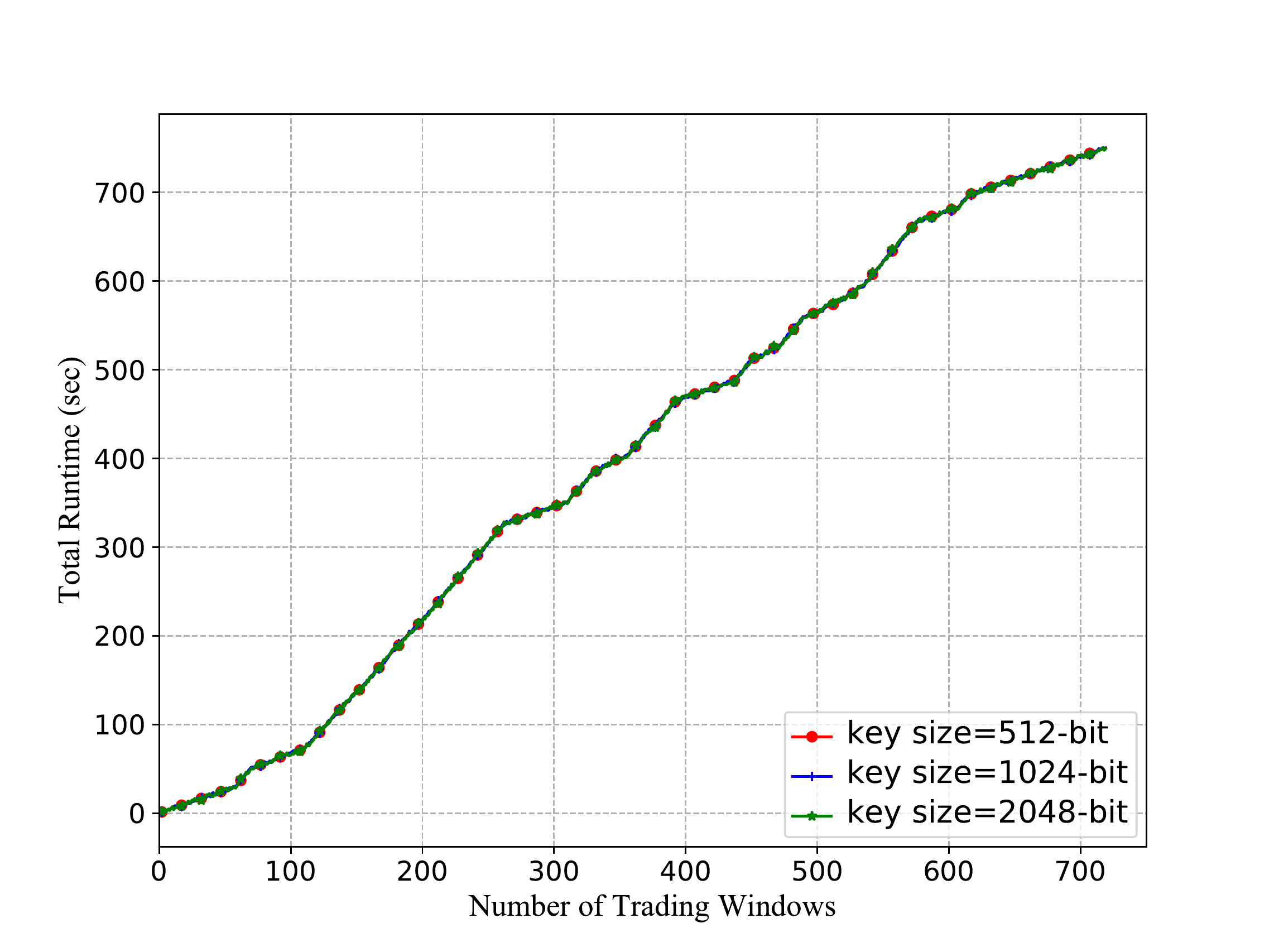}\hspace{-0.2in}
		\label{fig:1-2} }
\subfigure[Runtime (720 Trading Windows)]{
		\includegraphics[angle=0, width=0.32\linewidth]{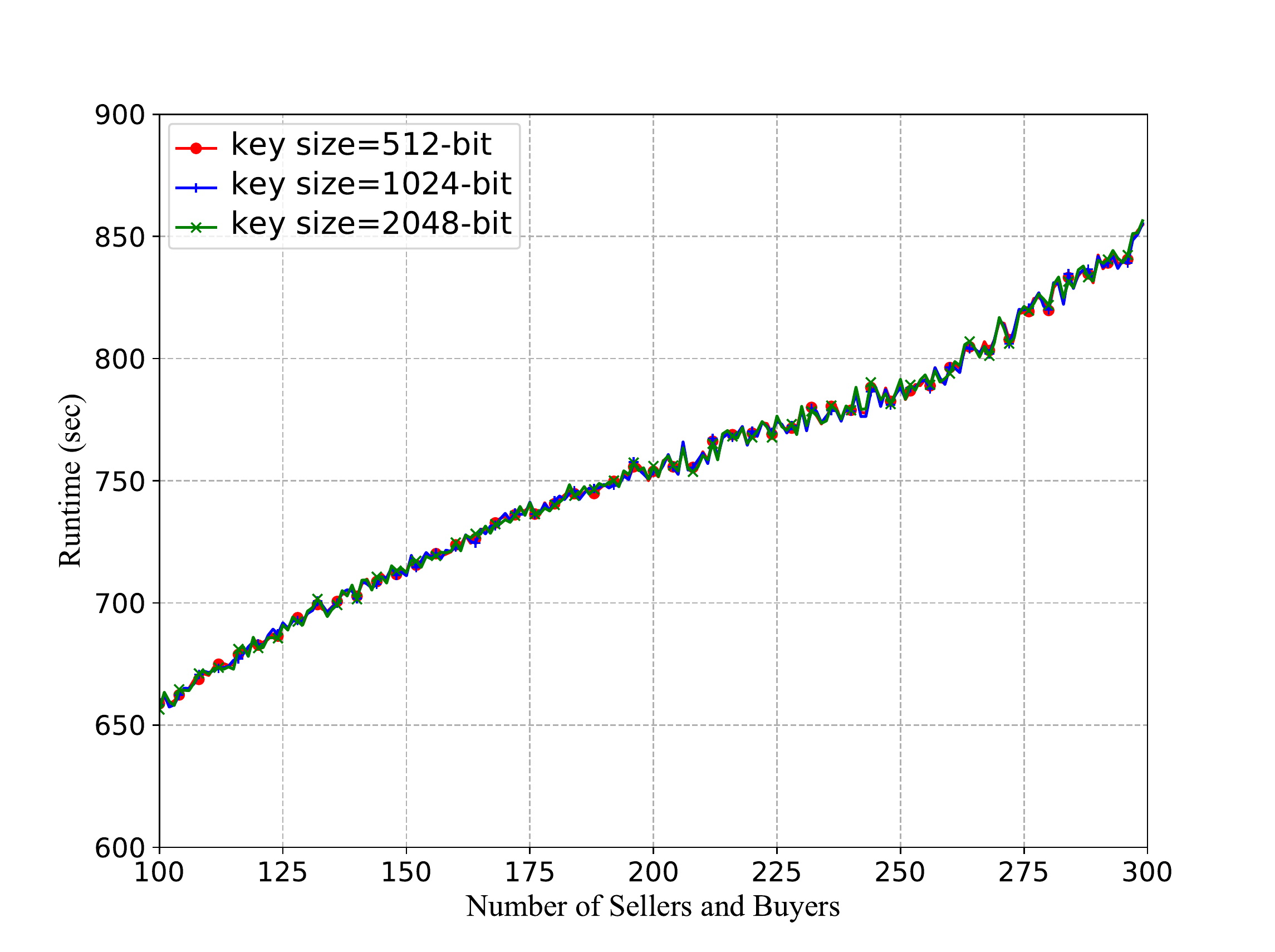}\label{fig:1-3} }
	\caption[Optional caption for list of figures]
	{Computational Performance Evaluation for the PEM Framework (negligible latency for minute-level inputs). 
	}
	\vspace{-0.1in}
	\label{fig:runtime}
\end{figure*}

Figure \ref{fig:size} illustrates the sizes of seller and buyer coalitions (the number of smart homes) in all 720 trading windows. The roles of smart homes change over time.

\subsection{Computational Performance Evaluation}

We evaluate the computational cost of PEM among 100 to 300 agents, using three different key sizes (512/1024/2048-bit). Fig. \ref{fig:1-1} shows the average runtime for a single trading window (including securely evaluating the market, computing the optimal price as well as the energy trading distribution amounts) as the number of trading windows varies from 1 to 720. The average runtime for each trading window is around 1 sec. This indicates that our PEM framework can efficiently function real-time trading in practice (\emph{with negligible latency}). 

Fig. \ref{fig:1-2} demonstrates the total runtime on different number of trading windows among 200 agents (with three different key sizes $512/1024/2048$-bit). Given the same number of trading windows, we observe that the key size for encryption and decryption executed in our protocols does not affect the runtime (since the encryption and decryption are independently executed in parallel during idle time). Finally, Fig. \ref{fig:1-3} validates that the total runtime increases as the number of agents increases.  

\subsection{Energy Trading Performance Evaluation}
We have also evaluated the trading performance of our PEM framework from the following perspectives:

\begin{itemize}
    \item the optimal price in all the trading windows;
       \item utility received by some representative sellers;
    \item total cost $\Gamma^t$ for the buyer coalition; 
    \item interactions with the main grid.
\end{itemize}
\subsubsection{Optimal Trading Price}

Fig. \ref{fig:price} shows the optimal prices in all the 720 trading windows. We can observe that the price changes over time: in the first few trading windows, the price equals $ps_g^t$ (purchasing all the energy from the grid). This shows that at the beginning of the day, the generation is close to 0, all the agents have to buy energy from the main grid. Similarly, at the end of day (around 7:00pm), the price is still $ps_g^t$ for the same reason. Furthermore, in many trading windows in the middle of the day, the trading price would be lower bounded: either the optimal price in the general market is out of range (this also applies to the upper bound), or the extreme market occurs. 

\subsubsection{Utility and Total Cost} We fix the preference parameter $k=20, 40$ for all the sellers in different trading windows.  
Fig. \ref{fig:utility} presents the utility of two agents (\emph{which are sellers in all 720 trading windows}).  
We have the following observations: 
\begin{itemize}
\item The utility of the agents with our PEM framework is higher than their utility without PEM (buyers only purchase energy from the main grid).

\item The utility improvement (with the PEM) in case of $k=40$ is higher than $k=20$. since lower preference parameter would make the sellers to sell more local energy (which results in more payoff).

\end{itemize}

In addition, Fig. \ref{fig:totalcost} shows the total cost of buyer coalition in the PEM (for 100 and 200 agents), which can be greatly reduced in all trading windows (e.g., 25.3\% in the current setting on average).

\subsubsection{Interaction with the Main Grid} Our PEM framework can also benefit the main grid by reducing the interactions between the agents and the grid, which is measured by the amount of electricity all the agents request from or feed into the grid. As shown in Fig. \ref{fig:impact}, since more energy can be traded in the PEM framework among agents, the interactions with the PEM are much lower than the original energy consumption (without the PEM).

\subsection{Communication Overheads}
We have also evaluated the bandwidth consumption of all the smart homes while executing the secure computation and communication among the 200 smart homes with different key sizes (512-bit, 1024-bit and 2048-bit). Table \ref{table:bandwidth_key} shows the average bandwidth over different numbers of trading windows (of all the smart homes). With such minor bandwidth consumption, our PEM framework can be deployed in most of the networking environments.

\begin{table}[!ht]
	\footnotesize
	\centering
	\caption{Average Bandwidth (MB) over $m$ Trading Windows}
		\begin{tabular}{c|ccccccccc}	
			\hline
		$m$	& 300& 360 & 420 & 480 & 540 & 600 & 660 & 720\\
			\hline
            512-bit &0.45&0.54&0.48 &0.52 &0.47 &0.48&0.55&0.46 \\
			1024-bit &0.84&0.88 &1.02 &0.93&0.98 & 1.06&0.97 &0.96 \\
			2048-bit &1.87&2.12 &2.05 &2.11&2.20&2.16&2.05&2.01 \\\hline
		\end{tabular}
	\label{table:bandwidth_key}
\end{table}

\begin{figure*}[!tbh]
	\centering
    \subfigure[Trading Price (200 smart homes)]{
		\includegraphics[angle=0, width=0.48\linewidth]{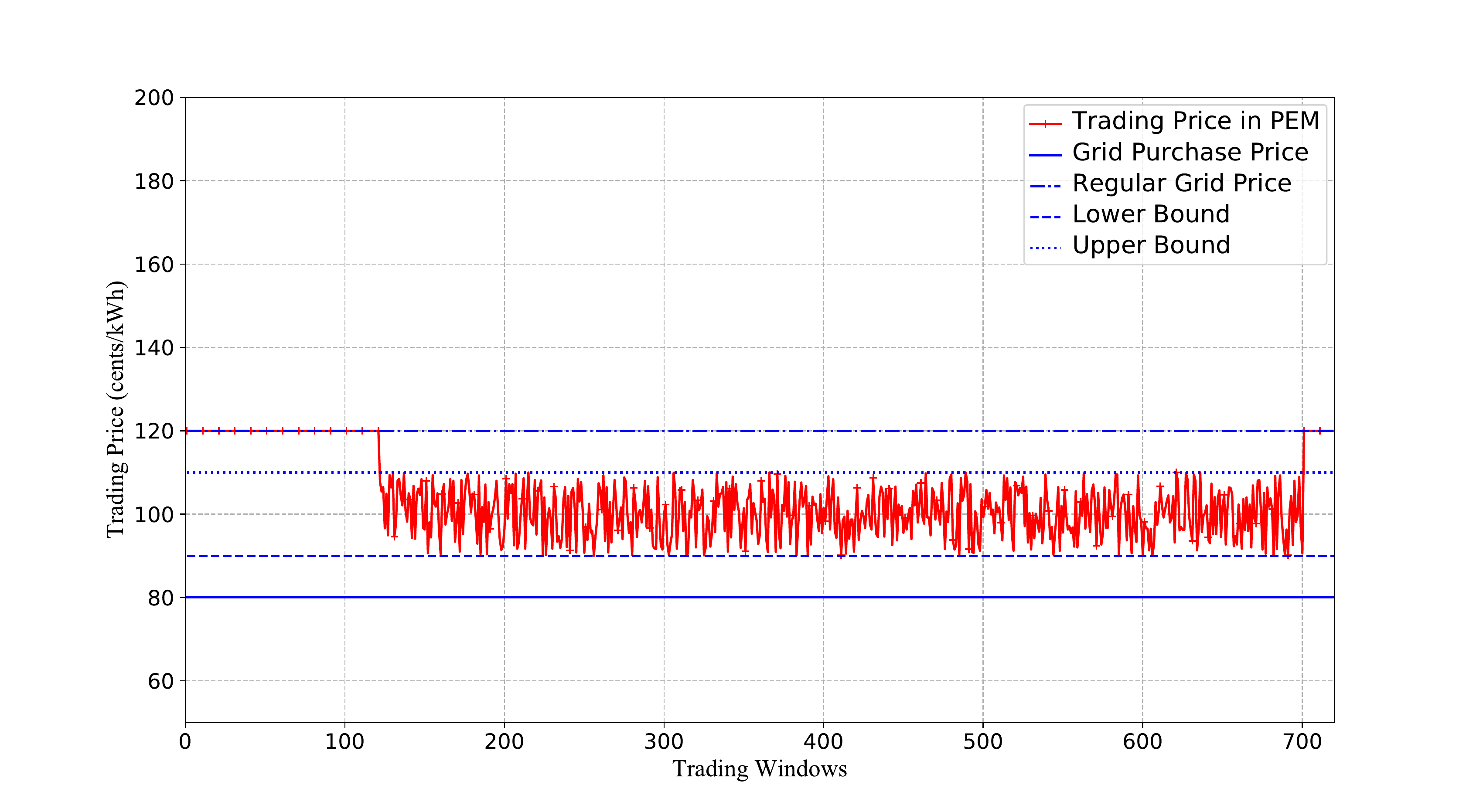}
		\label{fig:price} }
      \subfigure[Utility of Two Sellers]{
		\includegraphics[angle=0, width=0.48\linewidth]{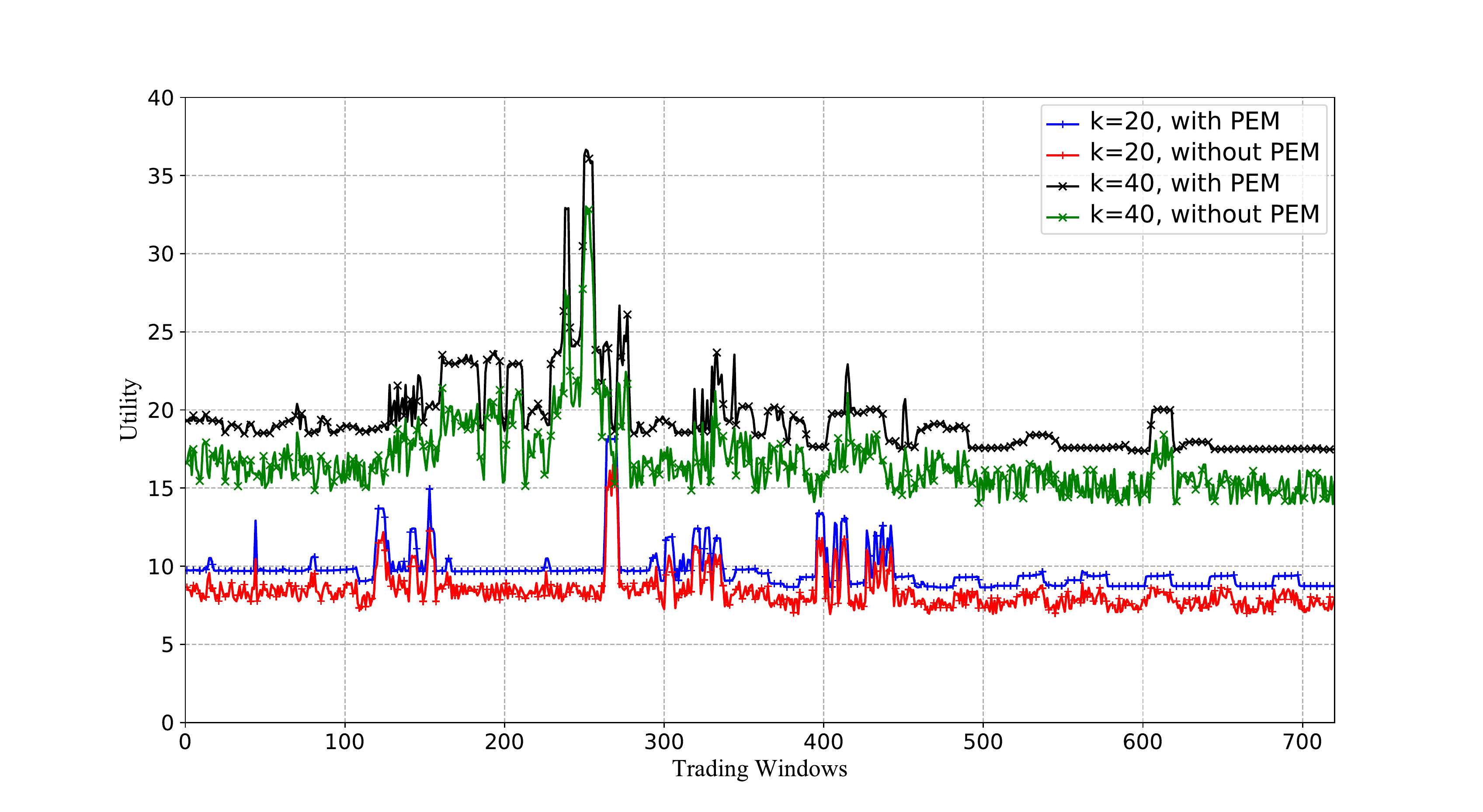}
		\label{fig:utility} }
	
	\subfigure[{Total Costs of Buyer Coalition}]{
		\includegraphics[angle=0, width=0.48\linewidth]{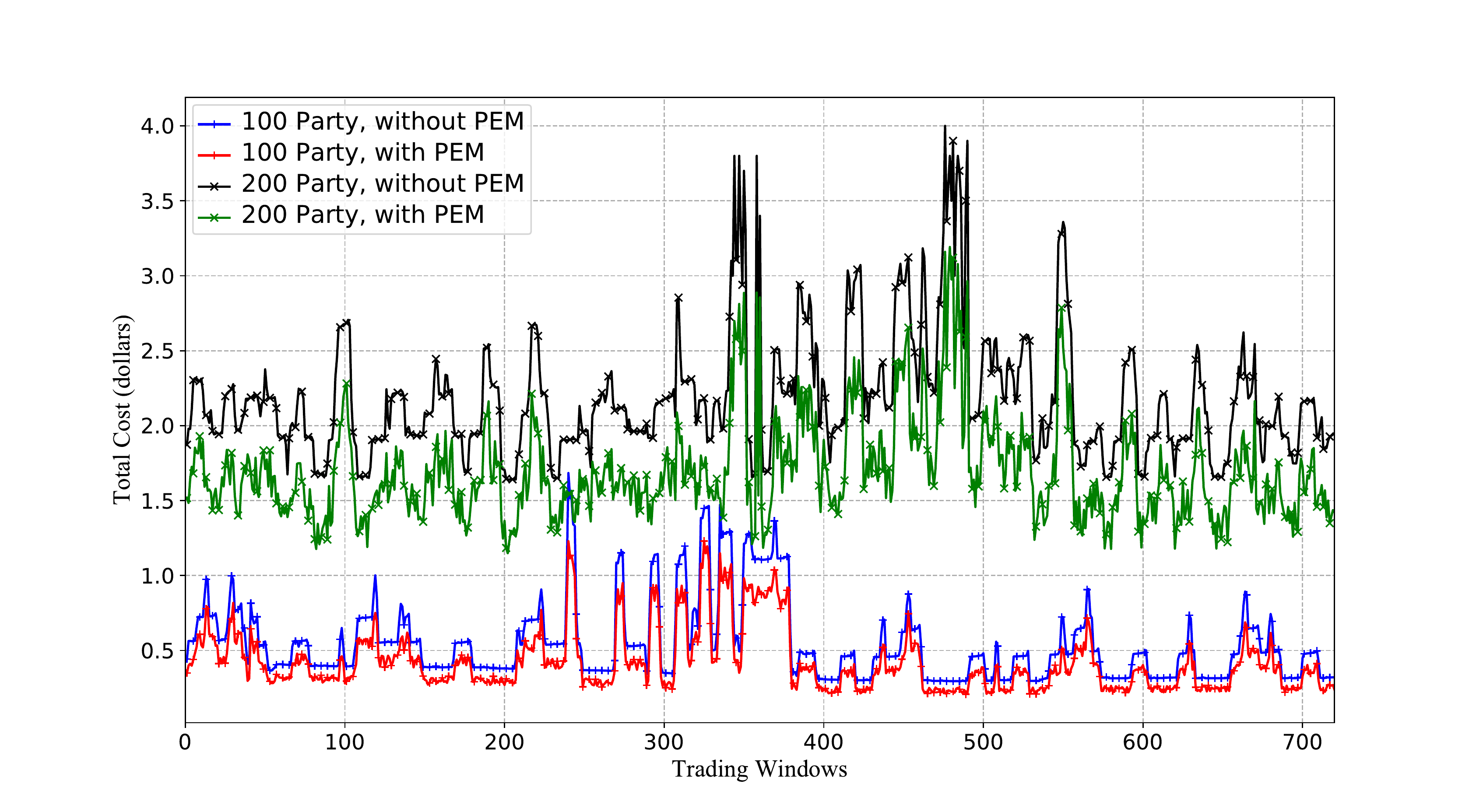}
	\label{fig:totalcost} }
	\subfigure[Interactions with Main Grid]{
		\includegraphics[angle=0, width=0.48\linewidth]{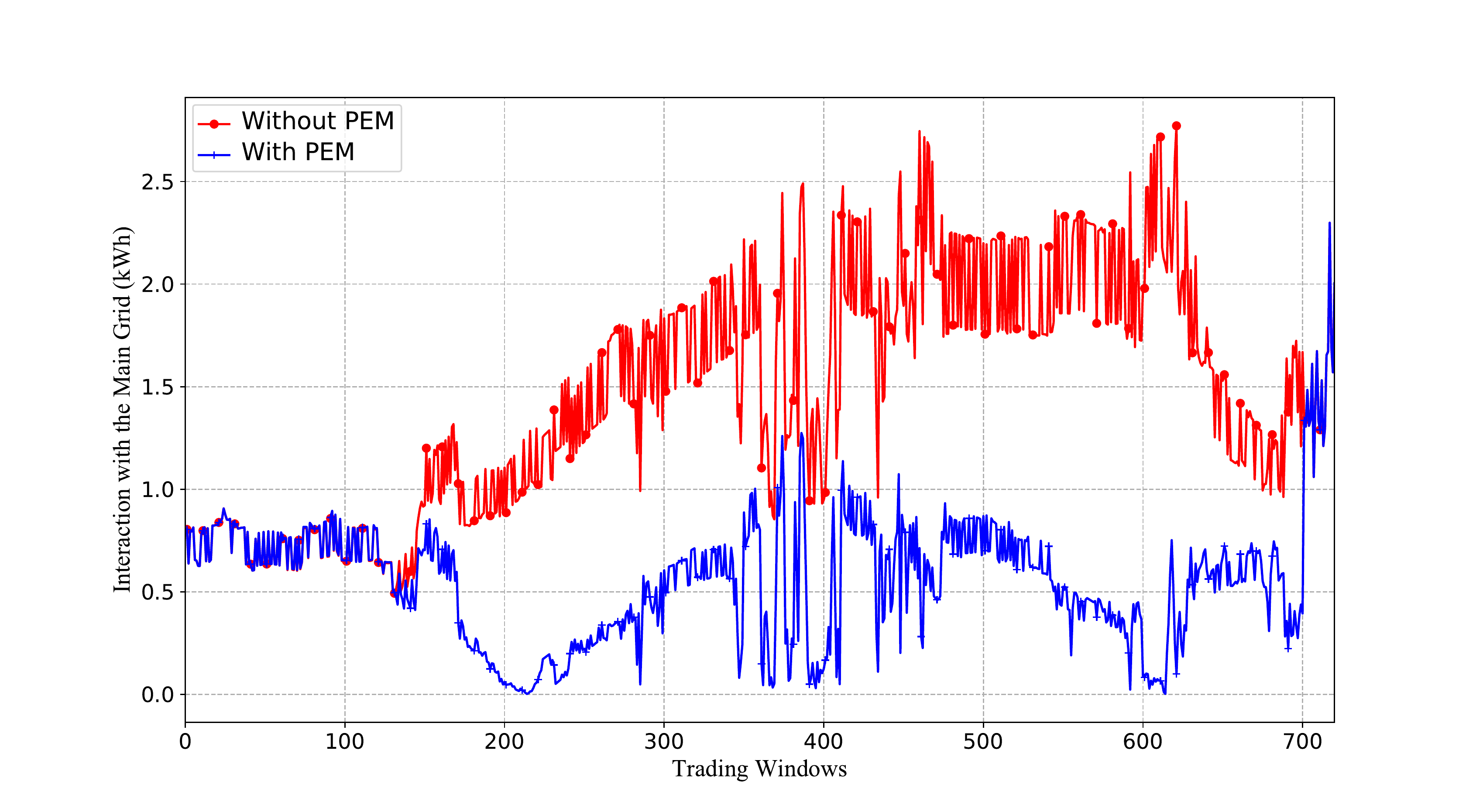}
	\label{fig:impact} }
	\caption[Optional caption for list of figures]
	{Energy Trading Performance Evaluation for the PEM Framework}\vspace{-0.1in}
\end{figure*}

\section{Related Work}
\label{sec:related}
\noindent\textbf{Smart Grid Privacy}. Most smart grid privacy research focuses on protecting data collected from smart meters integrated in the power grid \cite{MarmolSUP12}. Different privacy preserving techniques have been proposed to tackle such privacy concerns \cite{AcsC11,RottondiVC13}. For instance, He et al. \cite{Xinwen13} presented a distortion based privacy preserving metering scheme by introducing tolerable noise to obfuscate the consumption data. Rottondi et al. \cite{RottondiVC13} leveraged secure communication protocols to implement a privacy preserving infrastructure which allows utilities and data consumers to collect measurement data by securely aggregating smart metering data. \cite{ninghuiBat12} studied how to utilize the renewable energy sources (i.e., batteries) to hide the load/metering information of individual households.\cite{dimitriou2013privacy} proposed a privacy-preserving way to aggregate smart metering data for the billing of utility provider. Recently, \cite{XieHW19} researches on privately balancing the power load between the main grid and agents (microgrids).  However, none of such techniques can be applied to multiagent energy trading.
\vspace{0.03in}

\noindent\textbf{Secure Computation.} The theory of Secure multiparty computation (MPC) \cite{Yao86,Goldreich87} has significantly advanced the development of collaborative computation among multiple parties, which guarantees that functions can be securely computed with limited disclosure. Recently, secure computation has been intensively applied for privacy preserving system design in different contexts such as location-based services \cite{PopaBBL11}, and medical data analysis \cite{WangHZTWB15}. Moreover, Furukawa et al. \cite{furukawa2017high} have recently proposed a three-party secure computation against both semi-honest and malicious adversaries, which achieves low communication complexity and simple computation. Barak et al. \cite{MPCasService} have proposed the MPC-as-a-service concept and implemented an end-to-end system for large scale P2P secure computation with low bandwidth. 
\vspace{0.03in}

\noindent\textbf{Energy Trading}. Energy trading has been widely discussed with the development of smart grid. The integration of renewable energy sources has greatly motivated studies of energy market, e.g., incentive mechanisms for trading \cite{incentive18} and auction \cite{LiuAuction}, and multi-agent energy management \cite{MASStore17}, which could improve the stability and utility of the grid. Furthermore, distributed energy trading has been identified as a promising scheme for the energy market in \cite{WtushP2P,ZHANGp2pbid}. There are also many ongoing projects, e.g., LO3 Energy \cite{lo3energy} which focuses on the commercial energy trading to encourage residential units to trade with the neighborhoods. \cite{gai2019privacy} focuses on the energy trading by blockchain. However, the scalability of the proposed scheme is not very clear. To the best of our knowledge, we design and implement the first privacy preserving distributed energy trading framework.


\section{Conclusion}
\label{sec:concl}

We have proposed a novel privacy preserving distributed energy trading framework (PEM) which ensures privacy, individual rationality, and incentive compatibility. The optimal price for both sellers and buyers can reach a unique equilibrium in the modeled Stackelberg game. Moreover, we have designed novel cryptographic protocols for the entire PEM framework. Theoretical analyses are given to prove all the properties of the PEM framework. Finally, we have implemented a prototype for PEM, and conducted experiments to evaluate the system performance using real smart grid datasets. The experimental results (high computational efficiency, low bandwidth consumption and negligible latency) demonstrate that PEM can be readily integrated into the smart grid infrastructure.

\section*{Acknowledgments}
This work is partially supported by the National Science Foundation (NSF) under Grant No. CNS-1745894. The authors would like to thank the anonymous reviewers for their constructive comments.



\begin{thebibliography}{10}
\providecommand{\url}[1]{#1}
\csname url@samestyle\endcsname
\providecommand{\newblock}{\relax}
\providecommand{\bibinfo}[2]{#2}
\providecommand{\BIBentrySTDinterwordspacing}{\spaceskip=0pt\relax}
\providecommand{\BIBentryALTinterwordstretchfactor}{4}
\providecommand{\BIBentryALTinterwordspacing}{\spaceskip=\fontdimen2\font plus
\BIBentryALTinterwordstretchfactor\fontdimen3\font minus
  \fontdimen4\font\relax}
\providecommand{\BIBforeignlanguage}[2]{{
\expandafter\ifx\csname l@#1\endcsname\relax
\typeout{** WARNING: IEEEtranS.bst: No hyphenation pattern has been}
\typeout{** loaded for the language `#1'. Using the pattern for}
\typeout{** the default language instead.}
\else
\language=\csname l@#1\endcsname
\fi
#2}}
\providecommand{\BIBdecl}{\relax}
\BIBdecl

\bibitem{lo3energy}
``https://lo3energy.com/,'' 2019.

\bibitem{solarray}
``http://www.solarray.com/completepackages/,'' 2019.

\bibitem{AcsC11}
G.~{\'A}cs and C.~Castelluccia, ``I have a dream! (differentially private smart
  metering),'' in \emph{Information Hiding}, 2011, pp. 118--132.

\bibitem{AgrawalKV15}
\BIBentryALTinterwordspacing
P.~Agrawal, A.~Kumar, and P.~Varakantham, ``Near-optimal decentralized power
  supply restoration in smart grids,'' in \emph{AAMAS}, 2015. 
\BIBentrySTDinterwordspacing

\bibitem{al2011coordinating}
A.~T. Al-Awami and E.~Sortomme, ``Coordinating vehicle-to-grid services with
  energy trading,'' \emph{IEEE TSG}, vol.~3, no.~1, pp.
  453--462, 2011.

\bibitem{MPCasService}
\BIBentryALTinterwordspacing
A.~Barak, M.~Hirt, L.~Koskas, and Y.~Lindell, ``An end-to-end system for large
  scale p2p mpc-as-a-service and low-bandwidth mpc for weak participants,'' in
  \emph{ACM CCS}, 2018, pp. 695--712. 
\BIBentrySTDinterwordspacing

\bibitem{Barker2012}
S.~Barker, A.~Mishra, D.~Irwin, P.~Shenoy, and J.~Albrecht, ``{SmartCap:
  Flattening peak electricity demand in smart homes},'' \emph{PerCom}, 2012.

\bibitem{stackgame}
T.~Basar and G.~J. Olsder, \emph{{Dynamic Noncooperative Game Theory}}, 1999.

\bibitem{fairplaymp}
\BIBentryALTinterwordspacing
A.~Ben-David, N.~Nisan, and B.~Pinkas, ``Fairplaymp: a system for secure
  multi-party computation,'' in \emph{ACM CCS}, 2008, pp. 257--266. 
\BIBentrySTDinterwordspacing

\bibitem{Byali:2018:FSC:3243734.3243784}
\BIBentryALTinterwordspacing
M.~Byali, A.~Joseph, A.~Patra, and D.~Ravi, ``Fast secure computation for small
  population over the internet,'' in \emph{CCS}, 2018, pp. 677--694.
\BIBentrySTDinterwordspacing

\bibitem{CerquidesPR15}
\BIBentryALTinterwordspacing
J.~Cerquides, G.~Picard, and J.~A. Rodr{\'{\i}}guez{-}Aguilar, ``Designing a
  marketplace for the trading and distribution of energy in the smart grid,''
  in \emph{{AAMAS}}, 2015, pp. 1285--1293.
\BIBentrySTDinterwordspacing

\bibitem{dimitriou2013privacy}
T.~Dimitriou and G.~Karame, ``Privacy-friendly tasking and trading of energy in
  smart grids,'' in \emph{Proc. ASAC}, 2013, pp.
  652--659.

\bibitem{Fioretto0PMR17}
\BIBentryALTinterwordspacing
F.~Fioretto, W.~Yeoh, E.~Pontelli, Y.~Ma, and S.~J. Ranade, ``A distributed
  constraint optimization {(DCOP)} approach to the economic dispatch with
  demand response,'' in \emph{{AAMAS}}, 2017, pp. 999--1007.
\BIBentrySTDinterwordspacing

\bibitem{furukawa2017high}
J.~Furukawa, Y.~Lindell, A.~Nof, and O.~Weinstein, ``High-throughput secure
  three-party computation for malicious adversaries and an honest majority,''
  in \emph{Eurocrypt}, 2017, pp. 225--255.

\bibitem{gai2019privacy}
K.~Gai, Y.~Wu, L.~Zhu, M.~Qiu, and M.~Shen, ``Privacy-preserving energy trading
  using consortium blockchain in smart grid,'' \emph{IEEE Trans. on
  Industrial Informatics}, 2019.

\bibitem{garay2013rational}
J.~Garay, J.~Katz, U.~Maurer, B.~Tackmann, and V.~Zikas, ``Rational protocol
  design: Cryptography against incentive-driven adversaries,'' in \emph{IEEE FOCS}, 2013, pp. 648--657.

\bibitem{Gentry09}
\BIBentryALTinterwordspacing
C.~Gentry, ``Fully homomorphic encryption using ideal lattices,'' in
  \emph{{ACM} {STOC}}, 2009, pp. 169--178.
\BIBentrySTDinterwordspacing

\bibitem{Goldreichenc}
\BIBentryALTinterwordspacing
O.~Goldreich, \emph{The Foundations of Cryptography}.\hskip 1em plus 0.5em
  minus 0.4em\relax Cambridge University Press, 2004, vol.~2, ch. Encryption
  Schemes.
\BIBentrySTDinterwordspacing

\bibitem{Goldreich87}
\BIBentryALTinterwordspacing
O.~Goldreich, S.~Micali, and A.~Wigderson, ``How to play any mental game - a
  completeness theorem for protocols with honest majority,'' in
  \emph{ACM STOC}, 1987,
  pp. 218--229.
\BIBentrySTDinterwordspacing

\bibitem{Xinwen13}
X.~He, X.~Zhang, and C.-C.~J. Kuo, ``A distortion-based approach to
  privacy-preserving metering in smart grids,'' \emph{IEEE Practical
  Innovations: Open Solutions}, vol.~1, no.~3, pp. 67--78, 2013.

\bibitem{HongIJER15}
Y.~Hong, S.~Goel, and W.~M. Liu, ``An efficient and privacy-preserving scheme
  for p2p energy exchange among smart microgrids,'' \emph{International Journal
  of Energy Research}, vol.~40, no.~3, pp. 313--331, 2016.
  
\bibitem{HongESharing}
Y.~Hong, H.~Wang, S.~Xie, and B.~Liu, ``Privacy preserving and collusion resistant energy sharing,'' \emph{ICASSP}, 2018, pp. 6941--6945.  
 
\bibitem{HongLW17}
\BIBentryALTinterwordspacing
Y.~Hong, W.~M. Liu, and L.~Wang, ``Privacy preserving smart meter streaming
  against information leakage of appliance status,'' \emph{{IEEE} TIFS}, vol.~12, no.~9, pp. 2227--2241, 2017.
\BIBentrySTDinterwordspacing

\bibitem{kosba2016hawk}
A.~Kosba, A.~Miller, E.~Shi, Z.~Wen, and C.~Papamanthou, ``Hawk: The blockchain
  model of cryptography and privacy-preserving smart contracts,'' in \emph{IEEE SP}, 2016, pp. 839--858.

\bibitem{Kursawe:2011:PAS:2032162.2032172}
K.~Kursawe, G.~Danezis, and M.~Kohlweiss, ``Privacy-friendly aggregation for
  the smart-grid,'' in \emph{PETS'11}, 2011, pp. 175--191.


\bibitem{MaharjanZZGB13}
S.~Maharjan, Q.~Zhu, Y.~Zhang, S.~Gjessing, and T.~Basar, ``Dependable demand
  response management in the smart grid: A stackelberg game approach,''
  \emph{IEEE Trans. Smart Grid}, vol.~4, no.~1, pp. 120--132, 2013.

\bibitem{MalkhiNPS04}
\BIBentryALTinterwordspacing
D.~Malkhi, N.~Nisan, B.~Pinkas, and Y.~Sella, ``Fairplay - secure two-party
  computation system,'' in \emph{{USENIX} Security}, 2004.
\BIBentrySTDinterwordspacing

\bibitem{MarmolSUP12}
F.~G. M{\'a}rmol, C.~Sorge, O.~Ugus, and G.~M. P{\'e}rez, ``Do not snoop my
  habits: preserving privacy in the smart grid,'' \emph{IEEE Communications
  Magazine}, vol.~50, no.~5, pp. 166--172, 2012.

\bibitem{microeco}
A.~Mas-Colell, M.~Whinston, and J.~Green, \emph{{Microeconomic Theory}}, 1995.

\bibitem{mas1}
S.~D.~J. McArthur, E.~M. Davidson, V.~M. Catterson, A.~L. Dimeas, N.~D.
  Hatziargyriou, F.~Ponci, and T.~Funabashi, ``Multi-agent systems for power
  engineering applications—part i: Concepts, approaches, and technical
  challenges,'' \emph{IEEE TPS}, vol.~22, no.~4, pp.
  1743--1752, 2007.

\bibitem{mckenna2013photovoltaic}
E.~McKenna and M.~Thomson, ``Photovoltaic metering configurations, feed-in
  tariffs and the variable effective electricity prices that result,''
  \emph{IET Renewable Power Generation}, vol.~7, no.~3, pp. 235--245, 2013.

\bibitem{nakamoto2008bitcoin}
S.~Nakamoto \emph{et~al.}, ``Bitcoin: A peer-to-peer electronic cash system,''
  2008.

\bibitem{Nisan:2007:AGT:1296179}
N.~Nisan, T.~Roughgarden, E.~Tardos, and V.~V. Vazirani, \emph{Algorithmic Game
  Theory}.\hskip 1em plus 0.5em minus 0.4em\relax New York, NY, USA: Cambridge
  University Press, 2007.


\bibitem{LiuAuction}
B.~Liu, S.~Xie, and Y.~Hong, ``Privacy-Aware Double Auction for Divisible Resources without a Mediator,'' \emph{AAMAS}, 2020, to appear. 


\bibitem{MASStore17}
H.~Nunna and S.~Doolla, ``Multiagent-based
  distributed-energy-resource management for intelligent microgrids,''
  \emph{IEEE Trans. Industrial Electronics}, vol.~60, no.~4, pp.
  1678--1687, April 2013.

\bibitem{Paillier99}
P.~Paillier, ``Public key cryptosystems based on composite degree residuosity
  classes,'' in \emph{Advances in Cryptology - Eurocrypt '99 Proceedings, LNCS
  1592}, 1999, pp. 223--238.

\bibitem{PopaBBL11}
\BIBentryALTinterwordspacing
R.~A. Popa, A.~J. Blumberg, H.~Balakrishnan, and F.~H. Li, ``Privacy and
  accountability for location-based aggregate statistics,'' in
  \emph{{ACM} {CCS}}, 2011, pp.
  653--666. 
\BIBentrySTDinterwordspacing

\bibitem{RottondiVC13}
C.~Rottondi, G.~Verticale, and A.~Capone, ``Privacy-preserving smart metering
  with multiple data consumers,'' \emph{Computer Networks}, vol.~57, no.~7, pp.
  1699--1713, 2013.

\bibitem{UtilityFunc10}
P.~Samadi, A.~Mohsenian-Rad, R.~Schober, V.~W.~S. Wong, and J.~Jatskevich,
  ``Optimal real-time pricing algorithm based on utility maximization for smart
  grid,'' in \emph{IEEE SmartGridComm}, Oct 2010, pp. 415--420.

\bibitem{SankarRMP13}
L.~Sankar, S.~R. Rajagopalan, S.~Mohajer, and H.~V. Poor, ``Smart meter
  privacy: A theoretical framework,'' \emph{IEEE TSG}, vol.~4,
  no.~2, 2013.

\bibitem{singer2010budget}
Y.~Singer, ``Budget feasible mechanisms,'' in \emph{IEEE FOCS}, 2010.

\bibitem{TanGP13}
O.~Tan, D.~G{\"u}nd{\"u}z, and H.~V. Poor, ``Increasing smart meter privacy
  through energy harvesting and storage devices,'' \emph{IEEE Journal on
  Selected Areas in Communications}, vol.~31, no.~7, pp. 1331--1341, 2013.

\bibitem{sellenergy13}
C.~Tham and T.~Luo, ``Sensing-driven energy purchasing in smart grid
  cyber-physical system,'' \emph{IEEE Transactions on Systems, Man, and
  Cybernetics: Systems}, vol.~43, no.~4, pp. 773--784, July 2013.

\bibitem{Poor3party15}
W.~Tushar, B.~Chai, C.~Yuen, D.~B. Smith, K.~L. Wood, Z.~Yang, and H.~V. Poor,
  ``Three-party energy management with distributed energy resources in smart
  grid,'' \emph{IEEE Transactions on Industrial Electronics}, vol.~62, no.~4,
  pp. 2487--2498, April 2015.

\bibitem{WtushP2P}
W.~Tushar, C.~Yuen, H.~Mohsenian-Rad, T.~Saha, H.~V. Poor, and K.~L. Wood,
  ``Transforming energy networks via peer-to-peer energy trading: The potential
  of game-theoretic approaches,'' \emph{IEEE Signal Processing Magazine},
  vol.~35, no.~4, pp. 90--111, July 2018.

\bibitem{incentive18}
H.~Wang and J.~Huang, ``Incentivizing energy trading for interconnected
  microgrids,'' \emph{IEEE TSG}, vol.~9, no.~4, pp.
  2647--2657, 2018.

\bibitem{WangHZTWB15}
\BIBentryALTinterwordspacing
X.~S. Wang, Y.~Huang, Y.~Zhao, H.~Tang, X.~Wang, and D.~Bu, ``Efficient
  genome-wide, privacy-preserving similar patient query based on private edit
  distance,'' in \emph{ACM CCS},
  2015, pp. 492--503. 
\BIBentrySTDinterwordspacing

\bibitem{wang2014game}
Y.~Wang, W.~Saad, Z.~Han, H.~V. Poor, and T.~Ba{\c{s}}ar, ``A game-theoretic
  approach to energy trading in the smart grid,'' \emph{IEEE Transactions on
  Smart Grid}, vol.~5, no.~3, pp. 1439--1450, 2014.

\bibitem{wang2013active}
Z.~Wang, C.~Gu, F.~Li, P.~Bale, and H.~Sun, ``Active demand response using
  shared energy storage for household energy management,'' \emph{IEEE
  Transactions on Smart Grid}, vol.~4, no.~4, pp. 1888--1897, 2013.

\bibitem{Xie_TIFS19}
S.~{Xie}, Y.~{Hong}, and P.~{Wan}, ``Pairing: Privately balancing multiparty
  real-time supply and demand on the power grid,'' \emph{IEEE Transactions on
  Information Forensics and Security}, pp. 1--1, 2019.

\bibitem{XieHW19}
\BIBentryALTinterwordspacing
S.~Xie, Y.~Hong, and P.~Wan, ``A privacy preserving multiagent system for load
  balancing in the smart grid,'' in \emph{{AAMAS}}, 2019, pp. 2273--2275. 
\BIBentrySTDinterwordspacing

\bibitem{yang2012crowdsourcing}
D.~Yang, G.~Xue, X.~Fang, and J.~Tang, ``Crowdsourcing to smartphones:
  Incentive mechanism design for mobile phone sensing,'' in \emph{Mobicom}, 2012, pp. 173--184.

\bibitem{LeiToN}
L.~{Yang}, H.~{Kim}, J.~{Zhang}, M.~{Chiang}, and C.~W. {Tan}, ``Pricing-based
  decentralized spectrum access control in cognitive radio networks,''
  \emph{IEEE/ACM Trans. on Networking}, vol.~21, no.~2, pp. 522--535,
  2013.

\bibitem{ninghuiBat12}
\BIBentryALTinterwordspacing
W.~Yang, N.~Li, Y.~Qi, W.~Qardaji, S.~McLaughlin, and P.~McDaniel, ``Minimizing
  private data disclosures in the smart grid,'' in \emph{ACM CCS}, 2012, pp.
  415--427. 
\BIBentrySTDinterwordspacing

\bibitem{Yao86}
A.~C. Yao, ``How to generate and exchange secrets,'' in \emph{IEEE FOCS}, 1986, pp. 162--167.

\bibitem{XuLTE}
X.~{Yuan}, X.~{Qin}, F.~{Tian}, Y.~T. {Hou}, W.~{Lou}, S.~F. {Midkiff}, and
  J.~H. {Reed}, ``Coexistence between wi-fi and lte on unlicensed spectrum: A
  human-centric approach,'' \emph{IEEE JSAC}, vol.~35, no.~4, pp. 964--977, 2017.

\bibitem{ZHANGp2pbid}
\BIBentryALTinterwordspacing
C.~Zhang, J.~Wu, M.~Cheng, Y.~Zhou, and C.~Long, ``A bidding system for
  peer-to-peer energy trading in a grid-connected microgrid,'' \emph{Energy
  Procedia}, vol. 103, pp. 147 -- 152, 2016. 
\BIBentrySTDinterwordspacing

\end{thebibliography}
\end{document}